\documentclass[a4paper,11pt,oneside,reqno,final]{amsart}
\numberwithin{equation}{section}

\usepackage{graphicx}
\usepackage[latin1]{inputenc}
\usepackage{geometry}
\usepackage{latexsym}
\usepackage{amstext} \usepackage{amsmath} \usepackage{amssymb}
\usepackage{psfrag}
\usepackage{graphics}
\usepackage{color}

\renewcommand{\phi}{\varphi}

 2
 3
 4
 2

\newtheorem{theorem}{Theorem}

\newtheorem{proposition}{Proposition}
\begin{document}

\title[Exact solution of the six-vertex model]{Exact solution of the six-vertex model with domain wall boundary conditions. Critical line between disordered and antiferroelectric phases}

\author{Pavel Bleher}
\address{Department of Mathematical Sciences,
Indiana University-Purdue University Indianapolis,
402 N. Blackford St., Indianapolis, IN 46202, U.S.A.}
\email{bleher@math.iupui.edu}

\author{Thomas Bothner}
\address{Department of Mathematical Sciences,
Indiana University-Purdue University Indianapolis,
402 N. Blackford St., Indianapolis, IN 46202, U.S.A.}
\email{tbothner@iupui.edu}

\thanks{The first author is supported in part
by the National Science Foundation (NSF) Grant DMS-0969254.
A part of this work was done while the first author was visiting the Institute for Mathematical Sciences, National University of Singapore in 2012.}

\date{\today}


\begin{abstract}
In the present article we obtain the large $N$ asymptotics of the partition function $Z_N$ of the six-vertex model with domain wall boundary conditions on the critical line between the disordered and antiferroelectric phases. Using the weights $a=1-x,b=1+x,c=2,|x|<1$, we prove that, as $N\rightarrow\infty$, $Z_N=CF^{N^2}N^{1/12}\left(1+O(N^{-1})\right)$, where $F$ is given by an explicit expression in $x$ and the $x$-dependency in $C$ is determined. This result reproduces and improves the one given in the physics literature by Bogoliubov, Kitaev and Zvonarev \cite{BKZ}. Furthermore, we prove that the free energy exhibits an infinite order phase transition between the disordered and antiferroelectric phases. Our proofs are based on the large $N$ asymptotics for the underlying orthogonal polynomials which involve a non-analytical weight function, the Deift-Zhou nonlinear steepest descent method to the corresponding Riemann-Hilbert problem, and the Toda equation for the tau-function.
\end{abstract}
\maketitle


\section{Introduction and statement of the main result} We begin with the description of the model under consideration.
Given a square lattice in $\mathbb{Z}^2$ with $N^2$ vertices, we assign arrows along each edge obeying the following rule: At every vertex two arrows point in and two arrows point out. Such a rule is called the {\it ice-rule} and it only admits six possible arrow configurations, see Figure \ref{fig1}, hence the name of the model: the {\it six-vertex model} or the {\it model of two-dimensional ice}.
\begin{figure}[tbh]
  \begin{center}
  \includegraphics[width=8cm,height=6cm]{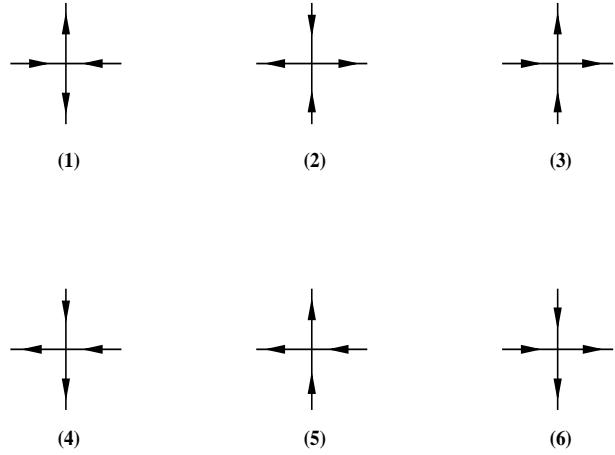}
  \end{center}
  \caption{The six allowed arrow configurations at a vertex}
  \label{fig1}
\end{figure}

\smallskip

On the boundary of the lattice we impose the {\it domain wall boundary conditions} (DWBC), that is all arrows on the top and bottom side of the lattice are directed inside the lattice and all arrows on the left and right side point outside, see Figure \ref{fig2} for a possible arrow configuration with DWBC.
\begin{figure}[tbh]
  \begin{center}
  \includegraphics[width=7cm,height=7cm]{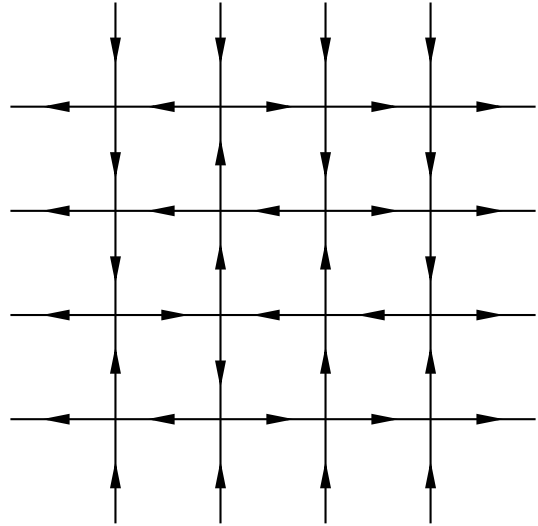}
  \end{center}
  \caption{An example of a possible $4\times 4$ configuration with DWBC}
  \label{fig2}
\end{figure}
\bigskip

To each of the six possible arrow configurations $(i)$ we assign a positive real-valued Boltzmann weight $w_i,i=1,\ldots,6$ and define the partition function $Z_N$ of the model as
\begin{equation*}
	Z_N=\sum_{\textnormal{configurations}}\prod_{i=1}^6w_i^{n_i}
\end{equation*} 
with $n_i$ denoting the number of vertices with arrow configuration $(i)$. By definition, $Z_N$ depends on six parameters: 
the weights $w_i$. However recalling some conservation laws (see for instance \cite{AR}, \cite{BL1} or \cite{FS}) 
we can reduce the number of parameters to two, namely,
\begin{equation}\label{conserv1}
	Z_N(w_1,w_2,w_3,w_4,w_5,w_6) = \bigg(\frac{w_5}{w_6}\bigg)^{N/2}Z_N(a,a,b,b,c,c),
\end{equation}
where
\begin{equation*}
	a=\sqrt{w_1w_2},\ \ b=\sqrt{w_3w_4},\ \ c=\sqrt{w_5w_6}\,,
\end{equation*}
and furthermore
\begin{equation}\label{conserv2}
	Z_N(a,a,b,b,c,c)=c^{N^2}Z_N\bigg(\frac{a}{c},\frac{a}{c},\frac{b}{c},\frac{b}{c},1,1\bigg).
\end{equation}
Thus, a general weight reduces to the two parameters $\frac{a}{c},\frac{b}{c}$. In order to study the phase diagram (see Figure \ref{fig3}) of the model, we introduce the parameter
\begin{equation*}
	\Delta = \frac{a^2+b^2-c^2}{2ab}\,.
\end{equation*}
\begin{figure}[tbh]
  \begin{center}
  \psfragscanon
  \psfrag{1}{\footnotesize{$0$}}
  \psfrag{2}{\footnotesize{$1$}}
  \psfrag{3}{\footnotesize{$1$}}
  \psfrag{4}{\footnotesize{$\frac{a}{c}$}}
  \psfrag{5}{\footnotesize{$\frac{b}{c}$}}
  \psfrag{6}{\footnotesize{$\textnormal{(AF)}$}}
  \psfrag{7}{\footnotesize{$\textnormal{(F)}$}}
  \psfrag{8}{\footnotesize{$\textnormal{(F)}$}}
  \psfrag{9}{\footnotesize{$\textnormal{(D)}$}}
  \includegraphics[width=6.5cm,height=6cm]{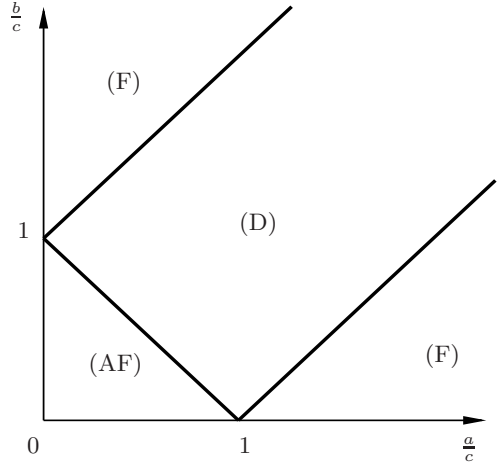}
  \end{center}
  \caption{The phase diagrom of the six-vertex model, with (F), (AF) and (D) denoting the relevant phases}
  \label{fig3}
\end{figure}
We distinguish the following three physical phase regions: the antiferroelectric phase (AF), $\Delta<-1$; the disordered phase (D), $-1<\Delta<1$; and, the ferroelectric phase (F), $\Delta>1$. In all phases the weights are usually parameterized in the following way: For the ferroelectric phase,
\begin{equation}\label{Fpara}
	(\textnormal{F})\hspace{0.5cm}a=\sinh(t-\gamma),\hspace{0.5cm} b=\sinh(t+\gamma),\hspace{0.5cm} c=\sinh(2\gamma),\hspace{0.75cm} 0<|\gamma|<t,
\end{equation}
for the antiferroelectric phase,
\begin{equation}\label{AFpara}
	(\textnormal{AF})\hspace{0.5cm}a=\sinh(\gamma-t),\hspace{0.5cm} b=\sinh(\gamma+t),\hspace{0.5cm} c=\sinh(2\gamma),\hspace{0.75cm} 0\leq|t|<\gamma
\end{equation}
and for the disordered phase,
\begin{equation}\label{Dpara}
	(\textnormal{D})\hspace{0.5cm}a=\sin(\gamma-t),\hspace{0.5cm} b=\sin(\gamma+t),\hspace{0.5cm} c=\sin(2\gamma),\hspace{0.75cm} 0\leq|t|<\gamma<\frac{\pi}{2}\,.
\end{equation}
Here we discuss the critical line between the disordered and antiferroelectric phase, hence
\begin{equation*}
	\Delta=-1
\end{equation*}
which corresponds to
\begin{equation*}
	\frac{a}{c}+\frac{b}{c}=1.
\end{equation*}
Instead of using the disordered phase parameterization \eqref{Dpara} in the limit $\gamma\rightarrow 0^+$, we choose the following parameterization for the weights $a,b$ and $c$:
\begin{equation}\label{abcpara}
	a=1-x,\ \ \ b=1+x,\ \ \ c=2,\hspace{0.5cm} x\in(-1,1).
\end{equation}
\smallskip

The above described six-vertex model with DWBC was introduced by Korepin in \cite{K} and further studied in \cite{I} and \cite{CIK}. Those works lead to a remarkable determinantal formula for the partition function with DWBC, which, with weights parameterized as
in \eqref{abcpara},  is
\begin{equation}\label{IKformula1}
	Z_N=\frac{(ab)^{N^2}}{(\prod_{k=0}^{N-1}k!)^2}\tau_N,
\end{equation}
where $\tau_N$ is the Hankel determinant,
\begin{equation}\label{IKformula2}
	\tau_N=\det\bigg(\frac{d^{i+j-2}}{dx^{i+j-2}}\phi(x)\bigg)_{i,j=1}^N,\hspace{0.5cm} \phi(x)=\frac{c}{ab}=\frac{2}{1-x^2}\,.
\end{equation}
The determinantal formula implies that $\tau_N$ solves the Toda equation,
\begin{equation}\label{toda}
	\tau_N\tau_N''-\big(\tau_N'\big)^2 = \tau_{N+1}\tau_{N-1},\ \ n\geq 1,\hspace{0.5cm} (')=\frac{\partial}{\partial x}\,.
\end{equation}
\smallskip

As was noticed by Zinn-Justin \cite{Z}, the Hankel determinant $\tau_N$ can be expressed
in terms of the partition function of a random matrix model, and then it can be reduced to orthogonal polynomials. 
On the critical line under consideration, the expression in terms of orthogonal polynomials can be obtained as follows. 
We write $\phi(x)$ as the Laplace transform of a continuous measure,
\begin{equation*}
	\phi(x)=\frac{2}{(1-x)(1+x)} = \int\limits_{-\infty}^{\infty}e^{x\lambda}m(\lambda)d\lambda,\hspace{0.5cm}m(\lambda)=e^{-|\lambda|}
\end{equation*}
and define the potential $V(\lambda)=|\lambda|-x\lambda$. This implies for $\tau_N$
\begin{equation*}	\tau_N=\int\limits_{-\infty}^{\infty}\cdots\int\limits_{-\infty}^{\infty}e^{-\sum_{k=1}^NV(\lambda_i)}\bigg(\prod_{l=1}^{N}\lambda_l^{l-1}\bigg)\det\Big(\lambda_i^{j-1}\big)_{i,j=1}^Nd\lambda_1\cdots d\lambda_N\\
\end{equation*}
and thus for any permutation $\sigma\in S_N$ acting on $\{\lambda_i\}_{i=1}^N$
\begin{equation*}
	\tau_N=\int\limits_{-\infty}^{\infty}\cdots\int\limits_{-\infty}^{\infty}
	e^{-\sum_{k=1}^NV(\lambda_i)}\bigg(\prod_{l=1}^N\textnormal{sgn}(\sigma)\lambda_{\sigma(l)}^{l-1}\bigg)\det\Big(\lambda_i^{j-1}\big)_{i,j=1}^Nd\lambda_1\cdots d\lambda_N,
\end{equation*}
hence after summation and identification of a Vandermonde determinant
\begin{equation}\label{tau1}
	\tau_N=\frac{1}{N!}\int\limits_{-\infty}^{\infty}\cdots\int\limits_{-\infty}^{\infty}e^{-\sum_{k=1}^NV(\lambda_k)}\prod_{i<j}(\lambda_i-\lambda_j)^2d\lambda_1\cdots d\lambda_N.
\end{equation}
Up to the factor $\frac{1}{N!}$, the expression on the right is the partition function of eigenvalues in the random matrix model with 
the potential $V$.

Introduce now monic orthogonal polynomials $\{p_n(s)\}_{n\geq 0}$ with respect to the measure $d\mu(s)=e^{-V(s)}ds$ on the real line,
\begin{equation*}
	\int\limits_{-\infty}^{\infty}p_n(s)p_m(s)d\mu(s) = h_n\delta_{nm}\,.
\end{equation*}
Then, by  using the orthogonality condition, one can simplify \eqref{tau1} to
\begin{equation}\label{connection}
	\tau_N=\prod_{k=0}^{N-1}h_k.
\end{equation}
Thus, for the partition function, via \eqref{IKformula1}, we obtain the formula
\begin{equation}\label{IKZformula}
	Z_N = (ab)^{N^2}\prod_{k=0}^{N-1}\frac{h_k}{(k!)^2}.
\end{equation}
We will prove the following asymptotics of the normalizing constants $h_N$:

\begin{theorem}\label{theo1} As $N\rightarrow\infty$,
\begin{equation}\label{theoasy1}
	\frac{h_N}{(N!)^2} = \bigg(\frac{\pi}{2\cos\frac{\pi x}{2}}\bigg)^{2N+1}\Big(1+\frac{1}{12N}+\varepsilon_N\Big),
\end{equation}
where
\begin{eqnarray}\label{theoasyerror1}
	\varepsilon_N&=&\frac{(-1)^N\cos(\pi x\big(N+\frac{1}{2}\big))}{2N(\ln N)^2}\bigg\{1-\frac{2\ln\ln N}{\ln N}+\frac{1-2\gamma-4\ln 2-2\ln\big(\cos\frac{\pi x}{2}\big)}{\ln N}\nonumber\\
	&&+O\Big(\bigg(\frac{\ln\ln N}{\ln N}\bigg)^2\Big)\bigg\}+O\big(N^{-2}\big),
\end{eqnarray}
and the error terms in \eqref{theoasyerror1} are uniform on any compact subset of the set
\begin{equation}\label{excset}
	\{ x\in\mathbb{R}:\ -1<x<1\}.
\end{equation}
\end{theorem}

The proof of Theorem \ref{theo1} is based on the Riemann-Hilbert approach for the potential $V$ with a singularity at the origin.
$V$ is close to the class of the Freud potentials considered in the paper of Kricherbauer and McLaughlin \cite{KM}. Our approach is somewhat different from \cite{KM}, and it gives a more detailed asymptotic formula for $h_N$, which is needed in the asymptotics
of $Z_N$.
Our second main result in the present paper is the asymptotics of $Z_N$.

\begin{theorem}\label{theo2} On the critical line between the disordered and antiferroelectric phase region with $-1<x<1$, as $N\rightarrow\infty$,
\begin{equation}\label{theoasy2}
	 Z_N= CF^{N^2}N^{\frac{1}{12}}\left(1+O\left(N^{-1}\right)\right),
\end{equation}
where
\begin{equation}\label{freeenergy}
	F=\frac{\pi(1-x^2)}{2\cos\frac{\pi x}{2}},\hspace{1cm} C=C_0\bigg(\cos\frac{\pi x}{2}\bigg)^{\frac{1}{12}}\,,
\end{equation}
and  $C_0>0$ is a universal constant. The error term in \eqref{theoasy2} is uniform on any compact subset of the set \eqref{excset}.
\end{theorem}

{\it Remark.} The leading term in formula \eqref{theoasy2} for $Z_N$, after a rescaling, coincides with the one obtained in the
physics literature by Bogoliubov, Kitaev and Zvonarev \cite{BKZ}. The error term estimate of
$O(N^{-1})$ is much stronger than a logarithmic estimate which can be derived from \cite{BKZ}.  Expression \eqref{freeenergy} for the constant $C$ is completely new.  

The current article is a continuation of the work of the first author with Vladimir Fokin in \cite{BF} and with Karl Liechty in \cite{BL1}, \cite{BL2} and \cite{BL3}. They prove conjectures of Paul Zinn-Justin in \cite{Z} on the large $N$ asymptotics of $Z_N$ in different phase regions: In the disordered phase with parameterization \eqref{Dpara} for some $\varepsilon>0$,
\begin{equation}\label{ZND}
	(\textnormal{D})\hspace{1.5cm} Z_N=CF^{N^2}N^{\kappa}\Big(1+O\big(N^{-\varepsilon}\big)\Big),\ \ N\rightarrow\infty
\end{equation}
with

\begin{equation}\label{disorderedenergy}
	\kappa=\frac{1}{12}-\frac{2\gamma^2}{3\pi(\pi-2\gamma)},\ \ \ \ F=\frac{\pi\big(\sin(\gamma-t)\sin(\gamma+t)\big)}{2\gamma\cos\frac{\pi t}{2\gamma}}
\end{equation}
and an $N$ independent constant $C>0$, whose value is unknown. In the antiferroelectric phase with \eqref{AFpara}
\begin{equation}\label{ZNAF}
	(\textnormal{AF})\hspace{1.5cm} Z_N=CF^{N^2}\vartheta_4(N\omega)\Big(1+O\big(N^{-1}\big)\Big),\quad N\rightarrow\infty,
\end{equation}
where
\begin{equation}\label{AFenergy}
	\omega=\frac{\pi}{2}\bigg(1+\frac{t}{\gamma}\bigg),\qquad
 F=\frac{\pi\big(\sinh(\gamma-t)\sinh(\gamma+t)\big)\vartheta'_1(0)}{2\gamma\vartheta_1(\omega)},
\end{equation}
and $\vartheta_1(z)=\vartheta_1(z|q)$, $\vartheta_4(z)=\vartheta_4(z|q)$ are the Jacobi theta functions 
with the elliptic nome $q=e^{-\frac{\pi^2}{2\gamma}}$. The constant $C>0$ is  also here unknown. 

Furthermore, in the ferroelectric phase \eqref{Fpara} for any $\varepsilon>0$
\begin{equation*}
	(\textnormal{F})\hspace{1.5cm} Z_N=CF^{N^2}G^N\Big(1+O\big(e^{-N^{1-\varepsilon}}\big)\Big),\ \ N\rightarrow\infty
\end{equation*}
with
\begin{equation*}
 	C=1-e^{-4\gamma},\ \ \ \ G=e^{\gamma-t},\ \ \ \ F=\sinh(\gamma +t).
\end{equation*} 	
Finally on the critical line between the ferroelectric and disordered phases with the parameterizations $a=\frac{1}{2}(t-1),b=\frac{1}{2}(t+1)$ and $c=1$ with $t>1$
\begin{equation*}
	(\textnormal{F-D})\hspace{1.5cm} Z_N=CF^{N^2}G^{\sqrt{N}}N^{\frac{1}{4}}\Big(1+O\big(N^{-1/2}\big)\Big),\ \ N\rightarrow\infty
\end{equation*}
where
\begin{equation*}
	F=\frac{t+1}{2},\ \ \ \ G=\exp\Bigg[-\zeta\bigg(\frac{3}{2}\bigg)\sqrt{\frac{t-1}{2\pi}}\Bigg]
\end{equation*}
with $\zeta(s)$ denoting the Riemann zeta function and a positive, yet unknown, constant $C>0$.

To calculate the order of the phase transition on the critical line between 
 the disordered and antiferroelectric phase regions, we would like to 
compare the free energy in these phase regions
as we approach the critical line. First we have to adjust the partition function in the regions.
Observe that 
as we approach the critical line, the parameters $t,\gamma\rightarrow 0$, hence $c\to 0$. 
We  rescale $Z_N$ according to \eqref{conserv2}:
\begin{equation*}
	Z_N\bigg(\frac{a}{c},\frac{a}{c},\frac{b}{c},\frac{b}{c},1,1\bigg)=c^{-N^2}Z_N(a,a,b,b,c,c).
\end{equation*}
In the disordered phase we obtain from  \eqref{ZND} that
\begin{equation}\label{ZND1}
	Z_N\bigg(\frac{a}{c},\frac{a}{c},\frac{b}{c},\frac{b}{c},1,1\bigg)
=CF^{N^2}N^{\kappa}\Big(1+O\big(N^{-\varepsilon}\big)\Big),\ \ N\rightarrow\infty,
\end{equation}
with the same $C$ and $\kappa$ as in \eqref{ZND}, and with
\begin{equation}\label{energyD}
(\textnormal{D})\hspace{1.5cm} 
	F=F(\gamma,t)=\frac{\pi ab}{2\gamma c\,\cos\frac{\pi t}{2\gamma}},\ \ \ 0\leq|t|<\gamma<\frac{\pi}{2}\,.
\end{equation}
Here $a,\,b,\,c$ are parameterized as in \eqref{Dpara}.

In the antiferroelectric region we obtain from  \eqref{ZNAF} that
\begin{equation}\label{ZNAF1}
Z_N\bigg(\frac{a}{c},\frac{a}{c},\frac{b}{c},\frac{b}{c},1,1\bigg)
=CF^{N^2}\vartheta_4(N\omega)\Big(1+O\big(N^{-1}\big)\Big),\ \ N\rightarrow\infty,
\end{equation}
with the same $C$ and $\omega$ as in \eqref{ZNAF}, and with
\begin{equation}\label{energyAF0}
	(\textnormal{AF})\hspace{1.5cm} 
	F=F(\gamma,t) = \frac{\pi ab\,\vartheta_1'(0)}{2\gamma c\,\vartheta_1(\omega)}
\,,\quad 0\leq|t|<\gamma.
\end{equation}
Here $a,\,b,\,c$ are parameterized as in \eqref{AFpara}. Let us remind that
\begin{equation}\label{Jac1}
\begin{aligned}
\vartheta_1(z)&=\vartheta_1(z|q)=2\sum_{n=0}^{\infty}(-1)^nq^{(n+\frac{1}{2})^2}\sin\big((2n+1)z\big),\qquad q=e^{-\frac{\pi^2}{2\gamma}},\\
	\omega&=\frac{\pi}{2}\bigg(1+\frac{t}{\gamma}\bigg).
\end{aligned}
\end{equation}
Observe that as $q\to 0$,
\begin{equation}\label{Jac2}
\begin{aligned}
\vartheta_1(z|q)=2q^{1/4}\sin z+O(q^{9/4}),\qquad \vartheta'_1(0|q)=2q^{1/4}+O(q^{9/4}),
\end{aligned}
\end{equation}
hence
\begin{equation}\label{jac2}
\begin{aligned}
\frac{\vartheta'_1(0)}{\vartheta_1(\omega)}=\frac{1}{\sin \omega}+O(q^2)=\frac{1}{\cos\frac{\pi t}{2\gamma}}+O(q^2).
\end{aligned}
\end{equation}
The free energy in the antiferroelectric region can be written as
\begin{equation}\label{decomp}
	F(\gamma,t) = F^{\textnormal{reg}}(\gamma,t)+F^{\textnormal{sing}}(\gamma,t),
\end{equation}
where
\begin{equation}\label{energyAF}
\begin{aligned}
	&F^{\textnormal{reg}}(\gamma,t)=\frac{\pi ab}{2\gamma c\,\cos\frac{\pi t}{2\gamma}},\\
 &a=\sinh(\gamma-t),
\quad b=\sinh(\gamma+t),\quad c=\sinh(2\gamma),
\end{aligned}
\end{equation}
and
\begin{equation*}
	F^{\textnormal{sing}}(\gamma,t)=\frac{\pi ab}{2\gamma c}\bigg(\frac{\vartheta_1'(0)}{\vartheta_1(\omega)}-\frac{1}{\cos\frac{\pi t}{2\gamma}}\bigg).
\end{equation*}
Observe at this point, that $F(\gamma,t)$ in the disordered phase region see \eqref{energyD}
and $F^{\textnormal{reg}}(\gamma,t)$ in the antiferroelectric one see \eqref{energyAF} are given by the same expression, 
but the underlying parameterizations of $a,b$ and $c$ are different.
\smallskip

Introduce now the coordinates $(x,y)$ on the phase diagram as
\begin{equation}\label{newcoordinates}
	\frac{a}{c}=\frac{1-x}{2}+y,\quad \frac{b}{c}=\frac{1+x}{2}+y,\quad  x\in(-1,1),
\end{equation}
so that $y>0$ corresponds to the disordered phase region, whereas $y<0$ to the antiferroelectric one. Employing the change of variables
\begin{equation}\label{changeofvar}
	(\gamma,t) \mapsto \big(x(\gamma,t),y(\gamma,t)\big),
\end{equation}
defined by \eqref{newcoordinates} and \eqref{AFpara} in the antiferroelectric phase region, and,
respectively, by \eqref{newcoordinates} and \eqref{Dpara} in the disordered phase region, we set
\begin{equation*}
	F_D(x,y) \equiv F\big(\gamma(x,y),t(x,y)\big)\; \textnormal{in (D)}, 
\hspace{0.75cm} F_{AF}(x,y) \equiv F\big(\gamma(x,y),t(x,y)\big)\;\textnormal{in (AF)},
\end{equation*}
and on the critical line (cf. Theorem \ref{theo2}),
\begin{equation*}
	F_C(x) \equiv F = \frac{\pi(1-x^2)}{4\cos\frac{\pi x}{2}}.
\end{equation*}
Also, we set in the (AF) region
\begin{equation}
	F_{AF}^{\rm reg}(x,y) \equiv F^{\rm reg}\big(\gamma(x,y),t(x,y)\big),\; 
\quad F_{AF}^{\rm sing}(x,y) \equiv F^{\rm sing}\big(\gamma(x,y),t(x,y)\big).
\end{equation}
We prove the following result on the order of the phase transition on the critical line, see also the work \cite{Z}
of Zinn-Justin.

\begin{theorem}\label{theo3} We have the following properties:
\begin{enumerate}
  \item  For any $x\in(-1,1)$,
change of coordinates \eqref{changeofvar} implies that as $y\rightarrow 0$,  
\begin{equation}\label{gt}
	\gamma = 2\sqrt{\frac{|y|}{1-x^2}}\big(1+O(y)\big),\hspace{1cm} t = 4x\sqrt{\frac{|y|}{1-x^2}}\big(1+O(y)\big).
\end{equation}
\item
The functions $F_D(x,y)$ and $F_{AF}^{\textnormal{reg}}(x,y)$ can be analytically continued in $y$ 
to a neighborhood of the point $y=0$, and in this neighborhood
\begin{equation}\label{FDAF}
	F_D(x,y)=F_{AF}^{\textnormal{reg}}(x,y).
\end{equation} 
In addition, we have that
\begin{equation}\label{FDAFC}
	F_D(0,x) = F_C(x)=F_{AF}^{\textnormal{reg}}(0,x).
\end{equation} 
\item
As $y\rightarrow 0$, $F_{AF}^{\textnormal{sing}}(x,y)$ is exponentially small,
\begin{equation}\label{singesti}
	F_{AF}^{\textnormal{sing}}(x,y) = O\Big(e^{-\frac{\pi^2}{\gamma}}\Big),\ \ x\in(-1,1),
\end{equation} 
that is we have an infinite order phase transition.
\end{enumerate}
\end{theorem}

The setup for the remainder of the paper is as follows: We will first proof Theorem \ref{theo1} within the framework of the Riemann-Hilbert approach to orthogonal polynomials. We start with a rescaling of the weight and the construction of the equilibrium measure. Afterwards a sequence of transformations is carried out which allow to approximate the global solution of the Riemann-Hilbert problem by local model functions, and enable us to solve the initial problem via iteration. After that, Theorem \ref{theo2} is a corollary to Theorem \ref{theo1}, via \eqref{IKZformula} and \eqref{toda}, with some extra arguments enabling us to prove the $O(N^{-1})$ estimate of the error term. The proof of Theorem \ref{theo3} will be given at the end of the article.


\section{Rescaling and the equilibrium measure}
Let us scale the variables in \eqref{tau1} as $\lambda_i=N\mu_i$, so that
\begin{equation}\label{tau2}
	\tau_N=\frac{N^{N^2}}{N!}\tilde{\tau}_N,\hspace{0.5cm} \tilde{\tau}_N=\int\limits_{-\infty}^{\infty}\cdots\int\limits_{-\infty}^{\infty}e^{-N\sum_{k=1}^NV(\mu_k)}\prod_{i<j}(\mu_i-\mu_j)^2d\mu_1\cdots d\mu_N.
\end{equation}
Now use the rescaled monic orthogonal polynomials $\{p_{N,n}(s)\}_{n\geq 0}$ with
\begin{equation*}
	p_{N,n}(s) = \frac{1}{N^n}p_n(Ns),
\end{equation*}
thus satisfying
\begin{equation*}
	\int\limits_{-\infty}^{\infty}p_{N,n}(s)p_{N,m}(s)e^{-NV(s)}ds = h_{N,n}\delta_{nm},\ \ \ \ h_{N,n}=\frac{h_n}{N^{2n+1}}
\end{equation*}
and simplify \eqref{tau2} similarly as in \eqref{connection}
\begin{equation*}
	\tilde{\tau}_N=N!\prod_{i=0}^{N-1}h_{N,i}.
\end{equation*}
Our original task of studying the asymptotics of orthogonal polynomials with respect to the weight function $w(s)=e^{-V(s)}$ has therefore simply been transfered to polynomials with the underlying weight
\begin{equation}\label{potential}
	\tilde{w}(s) = e^{-NV(s)},s\in\mathbb{R};\ \ \ \ \hspace{0.5cm} V(s)=|s|-xs,\ |x|<1.
\end{equation}
Our strategy will now focus on the determination of the large $N$ asymptotics of the normalizing constants $h_{N,N}$, and then, via $h_N=N^{2N+1}h_{N,N}$, the evaluation of $\tau_N$ from  \eqref{connection}.
\bigskip

To this end introduce the normalized counting measure $\nu$ on $\mathbb{R}$
\begin{equation*}
	\nu(s) = \frac{1}{N}\sum_{k=1}^{N}\delta(s-\mu_k),\hspace{0.5cm} \int\limits_{\mathbb{R}}d\nu(s)=1
\end{equation*}
and rewrite the integrand in \eqref{tau2} as
\begin{equation*}
	e^{-N\sum_{k=1}^NV(\mu_k)}\prod_{i<j}(\mu_i-\mu_j)^2=e^{-N^2H(\nu)}
\end{equation*}
with
\begin{equation*}
	H(\nu) = \int\int\ln|t-s|^{-1}d\nu(t)d\nu(s)+\int V(s)d\nu(s).
\end{equation*}
Since we are particularly interested in the behavior of $\tau_N$ as $N\rightarrow\infty$, we expect the value of $\tilde{\tau}_N$ in \eqref{tau2} to be focused in a neighborhood of the global minimum of the energy functional $H(\nu)$, with $\nu$ reaching over 
\begin{equation*}
	\mathcal{M}^1(\mathbb{R}) = \big\{\mu\in\textnormal{Borel measures on}\ \mathbb{R}:\int d\mu =1\big\}.
\end{equation*}
It is well known (for instance \cite{DKM}, \cite{DKMVZ}) that the energy minimization problem
\begin{equation*}
	E^V=\inf_{\mu\in\mathcal{M}^1(\mathbb{R})}\bigg[\int\int\ln|t-s|^{-1}d\mu(t)d\mu(s)+\int V(s)d\mu(s)\bigg].
\end{equation*}
has a unique solution $\mu = \mu^V\in\mathcal{M}^1(\mathbb{R})$, called the {\it equilibrium measure}.

In the given situation the underlying potential $ V(s)=|s|-xs$ is  convex. In this case the support
of the equilibrium measure consists of one interval,
\begin{equation}\label{eqsupport}
	J=\textnormal{supp}(\mu^V)=[\alpha,\beta]\subset\mathbb{R}.
\end{equation}
For our purposes the equilibrium measure will be essential in the construction of the {\it $g$-function}
\begin{equation}\label{gfunction}
	g(z)=\int\limits_J\ln(z-w)d\mu^V(w) = \int\limits_J\ln(z-w)\rho(w)dw,\hspace{0.5cm} z\in\mathbb{C}\backslash(-\infty,\beta_q],
\end{equation}
where we choose the principal branch for the logarithm. Moreover the $g$-function in turn determines the equilibrium measure uniquely by the following variational conditions: there exists a real constant $l$, the {\it Lagrange multiplier},  such that
\begin{equation}\label{identity1}
		z\in\mathbb{R}\backslash J:\ \ g_+(z)+g_-(z)-V(z)-l\leq 0,\ \ \ \ \ z\in J:\ \ g_+(z)+g_-(z)-V(z)-l=0.
\end{equation}
This characterization of the equilibrium measure leads to a Riemann-Hilbert problem for $g'(z)$ which can be solved as
\begin{equation}\label{gprimeRHP}
	g'(z) = \frac{\sqrt{R(z)}}{2\pi i}\int\limits_J \frac{V'(w)}{\sqrt{R(w)}_+}\frac{dw}{w-z},\ \ z\in\mathbb{C}\backslash J.
\end{equation}

It turns out that for the given potential,  $ V(s)=|s|-xs$, the $g$-function, its support, 
the density function $\rho(z)$, and the Lagrange multiplier $l$ can all be evaluated explicitly, see \cite{Z} and \cite{BF}. 
We have from \eqref{eqsupport}, \eqref{gprimeRHP} and \eqref{potential} that
\begin{equation*}
	g'(z) = \frac{\sqrt{(z-\alpha)(z-\beta)}}{2\pi i}\int\limits_{\alpha}^{\beta}\frac{\textnormal{sgn}(w)-x}{\sqrt{(w-\alpha)(w-\beta)}_+}\frac{dw}{w-z}\,.
\end{equation*}
To evaluate this integral we use the residue theorem,
\begin{eqnarray*}
	&&\int\limits_{\alpha}^{\beta}\frac{dw}{\sqrt{(w-\alpha)(w-\beta)}_+(w-z)} =\frac{1}{2}\int\limits_{C_0}\frac{dw}{\sqrt{(w-\alpha)(w-\beta)}(w-z)}\\
	&=&\pi i\ \textnormal{res}_{w=z}\frac{1}{\sqrt{(w-\alpha)(w-\beta)}(w-z)}+\pi i\ \textnormal{res}_{w=\infty}\frac{1}{\sqrt{(w-\alpha)(w-\beta)}(w-z)}\\
	&=&\frac{\pi i}{\sqrt{(z-\alpha)(z-\beta)}}
\end{eqnarray*}
where $C_0$ denotes a closed contour around the interval $[\alpha,\beta]$ such that both points $z$ and $\infty$ lie to the left of $C_0$. Similarly
\begin{equation*}
	\int\limits_{\alpha}^{\beta}\frac{\textnormal{sgn}(w)}{\sqrt{(w-\alpha)(w-\beta)}_+}\frac{dw}{w-z}=\frac{\pi i}{\sqrt{(z-\alpha)(z-\beta)}} +2i\int\limits_{\alpha}^0\frac{dw}{\sqrt{(w-\alpha)(\beta-w)}(w-z)}
\end{equation*}
and the latter integrand has an explicit antiderivative. We conclude
\begin{equation}\label{gprime}
	g'(z)=\frac{1-x}{2}+\frac{2}{\pi i}\ln\frac{\sqrt{\beta(z-\alpha)}-i\sqrt{-\alpha(z-\beta)}}{\sqrt{z(\beta-\alpha)}},\hspace{0.5cm}z\in\mathbb{C}\backslash[\alpha,\beta]
\end{equation}
choosing the principal branch for the square roots in \eqref{gprime}. In order to evaluate the endpoints $\alpha,\beta$ of the support, use the resolvent
 \begin{equation*}
	g'(z) = \int\limits_{\alpha}^{\beta}\frac{d\mu^V(w)}{z-w} = \frac{1}{z}+O\big(z^{-2}\big),\ \ z\rightarrow\infty
\end{equation*}
and compare with the expansion in \eqref{gprime} as $z\rightarrow\infty$. One obtains
\begin{equation}\label{endpoint}
\begin{aligned}
	&\alpha=-\pi\tan\frac{\pi}{4}(1-x),\qquad \beta=\pi\tan\frac{\pi}{4}(1+x),\\
&(-\alpha)\beta=\pi^2,\qquad \beta-\alpha=\frac{2\pi}{\cos\frac{\pi x}{2}}\,.
\end{aligned}
\end{equation}
For the $g$-function itself we use an indefinite integration by parts,
\begin{equation*}
	g(z) = zg'(z)-\int zg''(z)dz +C,\ \ \ \ C\equiv\textnormal{const},
\end{equation*}
where from \eqref{gprime}
\begin{equation*}
	g''(z)=\frac{\sqrt{(-\alpha)\beta}}{z\pi(\beta-\alpha)}\Bigg(\sqrt{\frac{z-\beta}{z-\alpha}}-\sqrt{\frac{z-\alpha}{z-\beta}}\Bigg).
\end{equation*}
Thus,
\begin{equation}\label{finalg}
	g(z)=zg'(z)+2\ln\big(\sqrt{z-\alpha}+\sqrt{z-\beta}\big)-(1+2\ln 2),\ \ \ z\in\mathbb{C}\backslash[\alpha,\beta],
\end{equation}
and the constant $C=-(1+2\ln 2)$ is obtained from a comparison of the asymptotics in 
\begin{equation}\label{gasyinfinity}
	g(z)=\int\limits_{\alpha}^{\beta}\ln(z-w)d\mu^V(w) = \ln z +O\big(z^{-1}\big),\ \ z\rightarrow\infty
\end{equation}
with the asymptotics in \eqref{finalg}. We move on to the density and use the resolvent equation again:
\begin{equation}\label{density1}
	g_+'(z)-g_-'(z)=-2\pi i\rho(z),\ \ z\in[\alpha,\beta].
\end{equation}
However, from \eqref{gprime},
\begin{equation*}
	g_{\pm}'(z) = \frac{\textnormal{sgn}(z)-x}{2}\pm\frac{2}{\pi i}\ln\frac{\sqrt{\beta(z-\alpha)}+\sqrt{-\alpha(\beta-z)}}{\sqrt{|z|(\beta-\alpha)}},\ \ z\in [\alpha,\beta],
\end{equation*}
since
\begin{equation*}
	\frac{\sqrt{\beta(z-\alpha)}-\sqrt{-\alpha(\beta-z)}}{\sqrt{|z|(\beta-\alpha)}} = \frac{\sqrt{|z|(\beta-\alpha)}}{\sqrt{\beta(z-\alpha)}+\sqrt{-\alpha(\beta-z)}}e^{i\textnormal{arg}z}.
\end{equation*}
Hence, via \eqref{density1},
\begin{equation}\label{density2}
	\rho(z)=\frac{2}{\pi^2}\ln\frac{\sqrt{\beta(z-\alpha)}+\sqrt{-\alpha(\beta-z)}}{\sqrt{|z|(\beta-\alpha)}},\ \ \ z\in(\alpha,\beta),
\end{equation}
thus the density has a logarithmic singularity at the origin. Finally, for the Lagrange multiplier, we state first, as a consequence of \eqref{density1} and \eqref{identity1}, that for $z\in[\alpha,\beta]$,
\begin{equation}\label{identity2}
	g_+(z)-g_-(z)=2\pi i\int\limits_z^{\beta}\rho(w)dw,\ \ \ g_+(z)+g_-(z)-V(z)-l=0.
\end{equation}
But for such $z$
\begin{eqnarray*}
	g_+(z) &=& \frac{V(z)}{2}+\frac{2z}{\pi i}\ln\frac{\sqrt{\beta(z-\alpha)}+\sqrt{-\alpha(\beta-z)}}{\sqrt{|z|(\beta-\alpha)}}+2\ln(\beta-\alpha)\\
	&&-2\ln\big(\sqrt{z-\alpha}-i\sqrt{\beta-z}\big)+C
\end{eqnarray*}
and
\begin{equation*}
	g_-(z) = \frac{V(z)}{2}-\frac{2z}{\pi i}\ln\frac{\sqrt{\beta(z-\alpha)}+\sqrt{-\alpha(\beta-z)}}{\sqrt{|z|(\beta-\alpha)}}+2\ln\big(\sqrt{z-\alpha}-i\sqrt{\beta-z}\big)+C,
\end{equation*}
hence
\begin{equation}\label{Lmultiplier}
	l = 2\ln(\beta-\alpha)-2(1+2\ln 2).
\end{equation}
We have now gathered enough information to start the asymptotical analysis.


\section{Riemann-Hilbert problem for orthogonal polynomials}

The following Riemann-Hilbert problem (RHP), originally introduced by Fokas, Its and Kitaev in \cite{FIK}, is essential: Find a $2\times 2$ piecewise analytic matrix-valued function $Y(z)=Y^{(n)}(z)$ such that
\begin{itemize}
	\item $Y^{(n)}(z)$ is analytic for $z\in\mathbb{C}\backslash \mathbb{R}$
	\item Orienting the real line from left to right the following jump relation holds
	\begin{equation*}
		Y^{(n)}_+(z)=Y^{(n)}_-(z)\begin{pmatrix}
		1 & \tilde{w}(z)\\
		0 & 1\\
		\end{pmatrix},\hspace{0.5cm}z\in\mathbb{R}
	\end{equation*}
	\item As $z\rightarrow\infty$, we have
	\begin{equation*}
		Y^{(n)}(z)=\Big(I+O\big(z^{-1}\big)\Big)z^{n\sigma_3},\ \ \ \ \sigma_3=\begin{pmatrix}
		1 & 0\\
		0 & -1\\
		\end{pmatrix}
	\end{equation*}
\end{itemize}
The stated problem has a unique solution $Y^{(n)}(z)$ given by
\begin{equation*}
	Y^{(n)}(z)=\begin{pmatrix}
	\pi_n(z) & \frac{1}{2\pi i}\int_{\mathbb{R}}\pi_n(s)\tilde{w}(s)\frac{ds}{s-z}\\
	\gamma_{n-1}\pi_{n-1}(z) & \frac{\gamma_{n-1}}{2\pi i}\int_{\mathbb{R}}\pi_{n-1}(s)\tilde{w}(s)\frac{ds}{s-z}\\
	\end{pmatrix}
\end{equation*}
where $\pi_n(z)=z^n+\ldots$ is the $n^{\textnormal{th}}$ monic orthogonal polynomial with respect to the measure $d\mu(s)=\tilde{w}(s)ds$ and 
\begin{equation*}
	\gamma_n = -\frac{2\pi i}{h_n},\hspace{0.5cm} h_n=\int\limits_{-\infty}^{\infty}\pi_n^2(s)d\mu(s).
\end{equation*}
Furthermore $Y^{(n)}(z)z^{-n\sigma_3}$ admits a full asymptotic expansion near infinity of the following form
\begin{equation*}
	Y^{(n)}(z)z^{-n\sigma_3} = I+\frac{Y_1^{(n)}}{z} +\frac{Y_2^{(n)}}{z^2}+O\big(z^{-3}\big),\ \ \ z\rightarrow\infty,\ \ Y_k^{(n)}=\big(Y_k^{(n)}\big)_{ij}
\end{equation*}
and since in our situation $\tilde{w}(s) = e^{-NV(s)}$, we obtain
\begin{equation}\label{constantRHP}
	h_{N,n}=-2\pi i\big(Y_1^{(n)}\big)_{12}
\end{equation}
which, in terms of the previous discussion in the last section, shows that we need to solve the given RHP for $Y(z) = Y^{(N)}(z)$. This solution will be obtained in the framework of the Deift-Zhou nonlinear steepest descent method \cite{DZ} using techniques developed in \cite{DKMVZ}. These techniques allow to approximate the global solution $Y(z)$ by local model functions, {\it parametrices}, and the iterative solution of a singular integral equation. We will elaborate the required steps in the following subsections.


\section{First transformation of the RHP - Normalization}
Recall \eqref{gasyinfinity} and make the following normalizing substitution in the orginal $Y$-RHP
\begin{equation}\label{firsttrafo}
	Y(z) = \exp\bigg(\frac{Nl}{2}\sigma_3\bigg)T(z)\exp\bigg(N\Big(g(z)-\frac{l}{2}\Big)\sigma_3\bigg),\ \ \ z\in\mathbb{C}\backslash\mathbb{R}.
\end{equation}
This leads to a RHP for the function $T(z)$
\begin{itemize}
		\item $T(z)$ is analytic for $z\in\mathbb{C}\backslash\mathbb{R}$.
		\item The above properties of $g(z)$, see \eqref{identity2}, imply the following jumps
		 \begin{equation*}
	T_+(z) = T_-(z)\begin{pmatrix}
	e^{-N(g_+-g_-)} & 1\\
	0 & e^{N(g_+-g_-)}\\
	\end{pmatrix},\ z\in[\alpha,\beta]
\end{equation*}
and
\begin{equation*}
	T_+(z) = T_-(z)\begin{pmatrix}
	1 & e^{N(g_++g_--V-l)}\\
	0 & 1\\
	\end{pmatrix},\ z\in\mathbb{R}\backslash[\alpha,\beta].
\end{equation*}
		\item At infinity, $T(z)$ is normalized
		\begin{equation*}
			T(z)=I+O\big(z^{-1}\big),\ \ z\rightarrow\infty.
		\end{equation*}
\end{itemize}
Let us take a closer look at the listed jump matrices: For $z\in(\beta,+\infty)$
\begin{equation*}
	g_+(z)+g_-(z)-V(z)-l = \frac{4z}{\pi}\textnormal{arg}\frac{\sqrt{\beta(z-\alpha)}-i\sqrt{-\alpha(z-\beta)}}{\sqrt{z(\beta-\alpha)}}+4\ln\frac{\sqrt{z-\alpha}+\sqrt{z-\beta}}{\sqrt{\beta-\alpha}}
\end{equation*}
and for $z\in(-\infty,\alpha)$
\begin{eqnarray*}
	g_+(z)+g_-(z)-V(z)-l &=& 2z +\frac{4z}{\pi}\textnormal{arg}\frac{\sqrt{\beta(\alpha-z)}-i\sqrt{-\alpha(\beta-z)}}{\sqrt{-z(\beta-\alpha)}}\\
	&&+4\ln\frac{\sqrt{\alpha-z}+\sqrt{\beta-z}}{\sqrt{\beta-\alpha}}
\end{eqnarray*}
thus (see \eqref{identity1})
\begin{equation*}
	g_+(z)+g_-(z)-V(z)-l<0,\ \ \ z\in\mathbb{R}\backslash[\alpha,\beta]
\end{equation*}
and we conclude for $z\in(-\infty,\alpha-\eta)\cup(\beta+\eta,+\infty),\eta>0$ fixed,
\begin{equation}\label{linebehavior}
	\begin{pmatrix}
	1 & e^{N(g_++g_--V-l)}\\
	0 & 1\\
	\end{pmatrix}\longrightarrow I,\ \ N\rightarrow\infty
\end{equation}
where the stated convergence is, in fact, exponentially fast. Secondly the line segment $[\alpha,\beta]$: The given potential $V(s)=|s|-xs,s\in\mathbb{R}$ admits analytical continuation in the left and right halfplane separately via
\begin{equation}\label{Vcontinued}
	V(z) = \left\{
                                   \begin{array}{ll}
                                     z(1-x), & \hbox{Re$z> 0$;} \\
                                     -z(1+x), & \hbox{Re$z< 0$,} 
                                   \end{array}
                                 \right.
\end{equation}
and this continuation is however two-valued on the imaginary axis. Since $g_-(z) = V(z)-g_+(z)+l,z\in[\alpha,\beta]$, the function
\begin{equation*}
	G(z) = g_+(z)-g_-(z) = 2g_+(z)-V(z)-l
\end{equation*}
admits analytical continuation in a neigborhood of the segment $[\alpha,\beta]$ into the first and second quadrant of the complex plane using the appropriate continuation of $V(z)$ in \eqref{Vcontinued}. On the other hand from \eqref{identity2}
\begin{equation}\label{capitalg}
	G(z) = 2\pi i\int\limits_z^{\beta}\rho(w)dw
\end{equation}
and therefore
\begin{equation*}
	\frac{d}{dy}G(z+iy)\Big|_{y=0}=2\pi\rho(z)=\frac{4}{\pi}\ln\frac{\sqrt{\beta(z-\alpha)}+\sqrt{-\alpha(\beta-z)}}{\sqrt{|z|(\beta-\alpha)}}>0,\ \ \ z\in(\alpha,0)\cup(0,\beta),
\end{equation*}
i.e. the stated local continuation of $G(z)$ into the first and second quadrant satisfies
\begin{equation}\label{realpart1}
	\textnormal{Re}\ G(z)>0\ \ \ \textnormal{for}\ \ \textnormal{Re}z\neq 0,\ \textnormal{Im}z>0.
\end{equation}
In the lower halfplane we argue in a similar fashion, indeed
\begin{equation*}
	G(z)=-2g_-(z)+V(z)+l
\end{equation*}
admits local continuation into the third and fourth quadrant satisfying
\begin{equation}\label{realpart2}
	\textnormal{Re}\ G(z)<0\ \ \ \textnormal{for}\ \ \textnormal{Re}z\neq 0,\ \textnormal{Im}z<0.
\end{equation}
These continuations will now be used in the following matrix factorizations
\begin{eqnarray*}
		\begin{pmatrix}
			e^{-N(g_+(z)-g_-(z))} & 1\\
			0 & e^{N(g_+(z)-g_-(z))}\\
		\end{pmatrix} &=& \begin{pmatrix}
		1 & 0\\
		e^{NG(z)} & 1\\
		\end{pmatrix}\begin{pmatrix}
		0 & 1\\
		-1 & 0\\
		\end{pmatrix}\begin{pmatrix}
		1 & 0\\
		e^{-NG(z)} & 1\\
		\end{pmatrix}\\
		 &=& S_{L_1}S_PS_{L_2},
\end{eqnarray*}
motivating the second transformation of the RHP.


\section{Second Transformation of the RHP - Opening of lenses}
Let $\mathcal{L}_j^{\pm}$ denote the {\it upper (lower) lense}, shown in Figure \ref{fig4}, which is bounded by the contour $\gamma_j^{\pm}$. Define
\begin{equation}\label{secondtrafo}
	S(z) = \left\{
                                 \begin{array}{ll}
                                  T(z)S_{L_2}^{-1} , & \hbox{$z\in\mathcal{L}_1^+\cup\mathcal{L}_2^+$,} \smallskip \\
                                  T(z)S_{L_1} , & \hbox{$z\in\mathcal{L}_3^-\cup\mathcal{L}_4^-$,} \smallskip\\
                                  T(z), & \hbox{otherwise,}
                                 \end{array}
                               \right.
\end{equation}
\begin{figure}[tbh]
  \begin{center}
  \psfragscanon
  \psfrag{1}{\footnotesize{$\alpha$}}
  \psfrag{2}{\footnotesize{$\beta$}}
  \psfrag{3}{\footnotesize{$\mathcal{L}_1^+$}}
  \psfrag{4}{\footnotesize{$\mathcal{L}_2^+$}}
  \psfrag{5}{\footnotesize{$\mathcal{L}_3^-$}}
  \psfrag{6}{\footnotesize{$\mathcal{L}_4^-$}}
  \psfrag{7}{\footnotesize{$\gamma_1^+$}}
  \psfrag{8}{\footnotesize{$\gamma_2^+$}}
  \psfrag{9}{\footnotesize{$\gamma_3^-$}}
  \psfrag{10}{\footnotesize{$\gamma_4^-$}}
  \includegraphics[width=10cm,height=3cm]{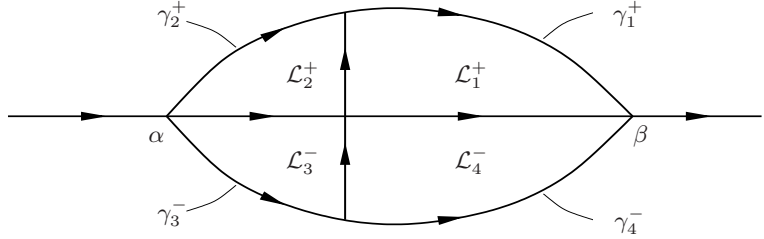}
  \end{center}
  \caption{The second transformation - opening of lenses}
  \label{fig4}
\end{figure}
then $S(z)$ solves the following RHP
\begin{itemize}
	\item $S(z)$ is analytic for $z\in\mathbb{C}\backslash(\mathbb{R}\cup \Gamma)$, with $\Gamma=\gamma_1^+\cup\gamma_2^+\cup\gamma_3^-\cup\gamma_4^-\cup(-i\varepsilon,i\varepsilon)$ and $0<\varepsilon<1$ remains fixed.
	\item The following jumps hold, with orientation fixed as in Figure \ref{fig4}
	\begin{equation*}
		S_+(z)=S_-(z)
	\left\{\begin{array}{ll}
                                  \bigl(\begin{smallmatrix}
                                   1 & 0\\
																e^{-NG(z)}& 1\\
																	\end{smallmatrix}\bigr), & \hbox{$z\in\gamma_1^+\cup\gamma_2^+$,} \smallskip \\
                                  \bigl(\begin{smallmatrix}
                                  0 & 1\\
                                  -1 & 0\\
                                  \end{smallmatrix}\bigr) , & \hbox{$z\in[\alpha,0)\cup(0,\beta]$,} \smallskip\\
                                  \bigl(\begin{smallmatrix}
                                  1 & e^{N(g_++g_--V-l)}\\
                                  0 & 1\\
                                  \end{smallmatrix}\bigr), & \hbox{$z\in\mathbb{R}\backslash[\alpha,\beta]$,} \smallskip\\
                                  \bigl(\begin{smallmatrix}
                                  1 & 0\\
                                  e^{NG(z)} & 1\\
                                  \end{smallmatrix}\bigr), & \hbox{$z\in\gamma_3^-\cup\gamma_4^-$.}\\
                                 \end{array}
                               \right.
   \end{equation*}
   Furthermore on the vertical line segment oriented upwards
   \begin{equation}\label{vertical1}
   	S_+(z)=S_-(z)\begin{pmatrix}
	1 & 0\\
	j_1(z) & 1\\
	\end{pmatrix},\ \ z\in(0,i\varepsilon)
\end{equation}
where
\begin{equation*}
	j_1(z) = e^{-NG_-(z)}-e^{-NG_+(z)}=-2ie^{-Nf_1(z)}\sin(Niz)
\end{equation*}
with 
\begin{eqnarray*}
	f_1(z)&=&\frac{4z}{\pi i}\ln\frac{\sqrt{\beta(z-\alpha)}+\sqrt{-\alpha(\beta-z)}}{\sqrt{\beta-\alpha}}-4\ln\frac{\sqrt{z-\alpha}-i\sqrt{\beta-z}}{\sqrt{\beta-\alpha}}\\
	&&+z-\frac{2z}{\pi i}\ln z
\end{eqnarray*}
and similarly
\begin{equation}\label{vertical2}
	S_+(z)=S_-(z)\begin{pmatrix}
	1 & 0\\
	j_2(z)& 1\\
	\end{pmatrix},\ \ z\in(-i\varepsilon,0)
\end{equation}
with
\begin{equation*}
	j_2(z) = e^{NG_+(z)}-e^{NG_-(z)} = 2ie^{-Nf_2(z)}\sin(Niz)
\end{equation*}
where
\begin{eqnarray*}
	f_2(z)&=&\frac{4(-z)}{\pi i}\ln\frac{\sqrt{\beta(z-\alpha)}+\sqrt{-\alpha(\beta-z)}}{\sqrt{\beta-\alpha}}+4\ln\frac{\sqrt{z-\alpha}-i\sqrt{\beta-z}}{\sqrt{\beta-\alpha}}\\
	&&-z-\frac{2(-z)}{\pi i}\ln(-z).
\end{eqnarray*}
\item At infinity, $S(z)=I+O\big(z^{-1}\big), z\rightarrow\infty$.
\end{itemize}
With \eqref{realpart1} and \eqref{realpart2}, the jump matrices will approach the identity matrix exponentially fast on the upper lense boundary $\gamma_1^+\cup\gamma_2^+$ and on the lower lense boundary $\gamma_3^-\cup\gamma_4^-$. Also, as we have seen earlier in \eqref{linebehavior}, a similar statement holds along $\mathbb{R}\backslash[\alpha,\beta]$. Let us direct our attention therefore towards the vertical line segment $(-i\varepsilon,i\varepsilon)$. For $z=iy,0<y<\varepsilon$ we have
\begin{equation}\label{finalvertical1}
	j_1(iy)=2ie^{-N(-\frac{2}{\pi}y\ln y+yh_1(y)+h_2(y))}\sin(Ny)
\end{equation}
where we introduced
\begin{equation*}
	h_1(y)=\frac{4}{\pi}\ln\frac{\sqrt{\beta(iy-\alpha)}+\sqrt{-\alpha(\beta-iy)}}{\sqrt{\beta-\alpha}},\ \ \ h_2(y)=4\ln\frac{\sqrt{iy-\alpha}+i\sqrt{\beta-iy}}{\sqrt{\beta-\alpha}}
\end{equation*}
and with $z=-iy,0<y<\varepsilon$ also
\begin{equation}\label{finalvertical2}
	j_2(-iy)=2ie^{-N(-\frac{2}{\pi}y\ln y+yh_3(y)+h_4(y))}\sin(Ny)
\end{equation}
with
\begin{equation*}
	h_3(y)=\frac{4}{\pi}\ln\frac{\sqrt{\beta(-iy-\alpha)}+\sqrt{-\alpha(\beta+iy)}}{\sqrt{\beta-\alpha}},\ \ \ h_4(y)=-4\ln\frac{\sqrt{-iy-\alpha}+i\sqrt{\beta+iy}}{\sqrt{\beta-\alpha}}.
\end{equation*}
Hence, combining \eqref{finalvertical1} and \eqref{finalvertical2}
\begin{equation*}
	j_1(z)=a_1(iy)e^{-\frac{2N}{\pi}y|\ln y|}\sin(Ny),\hspace{0.3cm} \ j_2(z)=a_2(-iy)e^{-\frac{2N}{\pi}y|\ln(y)|}\sin(Ny),\ \  0<y<\varepsilon
\end{equation*}
with functions $a_i(z)$ which are analytic in a full neighborhood of the origin, satisfyting $a_i(0)\neq 0$. As a result of an optimization consideration, the function
\begin{equation*}
	f(y) = e^{\frac{2N}{\pi}y\ln y}\sin(Ny),\ \ y\in[0,\varepsilon)
\end{equation*}
satisfies the important estimation
\begin{equation*}
	f(y) = O\bigg(\frac{1}{\ln N}\bigg),\ \ N\rightarrow\infty,\ \ y\in[0,\varepsilon),
\end{equation*}
hence also on the vertical line segment, the given jump matrices will, as $N\rightarrow\infty$, eventually approach the identity matrix. All together we expect, and this will be justified rigorously, that as $N\rightarrow\infty$, $S(z)$ converges to a solution of the model RHP, in which we only have to deal with the constant jump matrix on the punctured line segment $[\alpha,0)\cup(0,\beta]$. Let us now consider this model RHP.


\section{The Model RHP}

Find the piecewise analytic $2\times 2$ matrix valued function $M(z)$ such that
\begin{itemize}
	\item $M(z)$ is analytic for $z\in\mathbb{C}\backslash[\alpha,\beta]$
	\item Along $[\alpha,\beta]$, the following jump holds
	\begin{equation*}
		M_+(z) = M_-(z)\bigl(\begin{smallmatrix}
		\ 0 & 1\\
		-1 & \ 0\\
		\end{smallmatrix}\bigr),\hspace{0.5cm}z\in[\alpha,\beta]
	\end{equation*}
	\item $M(z)$ has at most logarithmic singularities at the endpoints $z=\alpha,\beta$
	\item $M(z)=I+O\big(z^{-1}\big),z\rightarrow\infty$
\end{itemize}
A solution to this problem can be obtained by diagonalization
\begin{equation}\label{permutationX}
	M(z)=\begin{pmatrix}
	-i & i\\
	1 & 1\\
	\end{pmatrix}\delta(z)^{-\sigma_3}\frac{i}{2}\begin{pmatrix}
	1 & -i\\
	-1 & -i\\
	\end{pmatrix}=\frac{1}{2}\begin{pmatrix}
	\delta+\delta^{-1} & i(\delta-\delta^{-1})\\
	-i(\delta-\delta^{-1}) & \delta+\delta^{-1}\\
	\end{pmatrix}
\end{equation}
with
\begin{equation*}
	\delta(z) = \bigg(\frac{z-\alpha}{z-\beta}\bigg)^{1/4}
\end{equation*}
defined on $\mathbb{C}\backslash[\alpha,\beta]$ with its branch fixed by the condition $\big(\frac{z-\alpha}{z-\beta}\big)^{1/4}\rightarrow 1$ as $z\rightarrow\infty$.


\section{Construction of a parametrix at the edge point $z=\beta$}

Fix a small neighborhood $\mathcal{U}$ of the point $\beta$ and observe that
\begin{eqnarray*}
	g_+(z)+g_-(z)-V(z)-l &=&\frac{4z}{\pi}\ln\frac{\sqrt{\beta(z-\alpha)}-i\sqrt{-\alpha(z-\beta)}}{\sqrt{z(\beta-\alpha)}}\\
	&&+4\ln\frac{\sqrt{z-\alpha}+\sqrt{z-\beta}}{\sqrt{\beta-\alpha}}\\
	&=&-c_0(z-\beta)^{3/2}+O\big((z-\beta)^2\big),\ \ \ c_0=\frac{8}{3\beta\sqrt{\beta-\alpha}}>0
\end{eqnarray*} 
as $z\in\mathcal{U},z>\beta$. Simultaneously,
\begin{equation*}
	G(z)=2g(z)-V(z)-l = -c_0(z-\beta)^{3/2}+O\big((z-\beta)^2\big),
\end{equation*}
as $z\in\mathcal{U}\cap \gamma_1^+$ and
\begin{equation*}
	G(z)=-2g(z)+V(z)+l=c_0(z-\beta)^{3/2}+O\big((z-\beta)^2\big),\ \ z\in\mathcal{U}\cap \gamma_4^-
\end{equation*}
where the function $(z-\beta)^{3/2}$ is defined for $z\in\mathbb{C}\backslash(-\infty,\beta]$ and fixed by the condition
\begin{equation*}
	(z-\beta)^{3/2}>0\ \ \textnormal{if}\ z>\beta.
\end{equation*}
The stated local behaviors suggest to use the Airy function $\textnormal{Ai}(\zeta)$ in the construction of an edge point parametrix. To this end first recall (see for instance \cite{BE}) that the function $\textnormal{Ai}(\zeta)$ is a solution to the Airy equation
\begin{equation*}
	w''=zw
\end{equation*}
uniquely fixed by its asymptotics as $\zeta\rightarrow\infty$ and $-\pi<\textnormal{arg}\ \zeta<\pi$
\begin{equation*}
	\textnormal{Ai}(\zeta)=\frac{\zeta^{-1/4}}{2\sqrt{\pi}}e^{-\frac{2}{3}\zeta^{3/2}}\Big(1-\frac{5}{48}\zeta^{-3/2}+\frac{385}{4608}\zeta^{-6/2}+O\big(\zeta^{-9/2}\big)\Big)
\end{equation*}
as well as for $-\frac{5\pi}{3}<\textnormal{arg}\ \zeta<-\frac{\pi}{3}$
\begin{eqnarray*}
	\textnormal{Ai}(\zeta) &=&\frac{\zeta^{-1/4}}{2\sqrt{\pi}}e^{-\frac{2}{3}\zeta^{3/2}}\Big(1-\frac{5}{48}\zeta^{-3/2}+\frac{385}{4608}\zeta^{-6/2}+O\big(\zeta^{-9/2}\big)\Big)\\
	&&-\frac{i\zeta^{-1/4}}{2\sqrt{\pi}}e^{\frac{2}{3}\zeta^{3/2}}\Big(1+\frac{5}{48}\zeta^{-3/2}+\frac{385}{4608}\zeta^{-6/2}+O\big(\zeta^{-9/2}\big)\Big),\ \ \zeta\rightarrow\infty.
\end{eqnarray*}
Now introduce for $\zeta\in\mathbb{C}$
\begin{equation}\label{A0}
	A_0(\zeta) = \begin{pmatrix}
	\frac{d}{d\zeta}\textnormal{Ai}(\zeta) & e^{i\frac{\pi}{3}}\frac{d}{d\zeta}\textnormal{Ai}\Big(e^{-i\frac{2\pi}{3}}\zeta\Big)\\
	\textnormal{Ai}(\zeta) & e^{i\frac{\pi}{3}}\textnormal{Ai}\Big(e^{-i\frac{2\pi}{3}}\zeta\Big)\\
	\end{pmatrix}
\end{equation}
and observe from the previously stated asymptotics that as $\zeta\rightarrow\infty$ for $0<\textnormal{arg}\ \zeta<\pi$
\begin{eqnarray*}
	A_0(\zeta) &=& \frac{\zeta^{\sigma_3/4}}{2\sqrt{\pi}}\begin{pmatrix}
	-1 & i\\
	1 & i\\
	\end{pmatrix}\Bigg[I+\frac{1}{48\zeta^{3/2}}\begin{pmatrix}
	1 & 6i\\
	6i & -1\\
	\end{pmatrix}+\frac{35}{4608\zeta^{6/2}}\begin{pmatrix}
	-1 & 12i\\
	-12i & -1\\
	\end{pmatrix}\\
	&&+O\big(\zeta^{-9/2}\big)\Bigg]e^{-\frac{2}{3}\zeta^{3/2}\sigma_3}.
\end{eqnarray*}
On the other hand if $-\pi<\textnormal{arg}\ \zeta<0$, we have
\begin{eqnarray*}
	A_0(\zeta) &=& \frac{\zeta^{\sigma_3/4}}{2\sqrt{\pi}}\begin{pmatrix}
	-1 & i\\
	1 & i\\
	\end{pmatrix}\Bigg[I+\frac{1}{48\zeta^{3/2}}\begin{pmatrix}
	1 & 6i\\
	6i & -1\\
	\end{pmatrix}
	+\frac{35}{4608\zeta^{6/2}}\begin{pmatrix}
	-1 & 12i\\
	-12i & -1\\
	\end{pmatrix}\\
	&&+O\big(\zeta^{-9/2}\big)\Bigg]e^{-\frac{2}{3}\zeta^{3/2}\sigma_3}\begin{pmatrix}
	1 & 1\\
	0 & 1\\
	\end{pmatrix},\ \ \ \zeta\rightarrow\infty.
\end{eqnarray*}
We assemble the following model function
\begin{equation}\label{ARH}
		A^{RH}(\zeta) = \left\{
                                 \begin{array}{ll}
                                   A_0(\zeta), & \hbox{arg $\zeta\in(0,\frac{2\pi}{3})$,} \smallskip\\
                                   A_0(\zeta)\begin{pmatrix}
                          1 & 0 \\
                          -1 & 1 \\
                        \end{pmatrix}, & \hbox{arg $\zeta\in(\frac{2\pi}{3},\pi)$,} \bigskip \\
                                   A_0(\zeta)\begin{pmatrix}
                                   1 & -1\\
                                   0 & 1\\
                                   \end{pmatrix},& \hbox{arg $\zeta\in(-\frac{2\pi}{3},0)$,} \bigskip \\
                                   A_0(\zeta)\begin{pmatrix}
                                   0 & -1 \\
                                   1 & 1\\
                                   \end{pmatrix}, & \hbox{arg $\zeta\in(-\pi,-\frac{2\pi}{3})$.}
                                 \end{array}
                               \right.
\end{equation}
leading to the RHP depicted in Figure \ref{fig5}
\begin{figure}[tbh]
  \begin{center}
  \psfragscanon
  \psfrag{1}{\footnotesize{$\begin{pmatrix}
  1 & 1\\
  0 & 1\\
  \end{pmatrix}$}}
  \psfrag{2}{\footnotesize{$\begin{pmatrix}
  1 & 0\\
  1 & 1\\
  \end{pmatrix}$}}
  \psfrag{3}{\footnotesize{$\begin{pmatrix}
  1 & 0\\
  1 & 1\\
  \end{pmatrix}$}}
  \psfrag{4}{\footnotesize{$\begin{pmatrix}
  0 & 1\\
  -1 & 0\\
  \end{pmatrix}$}}
  \includegraphics[width=5cm,height=4cm]{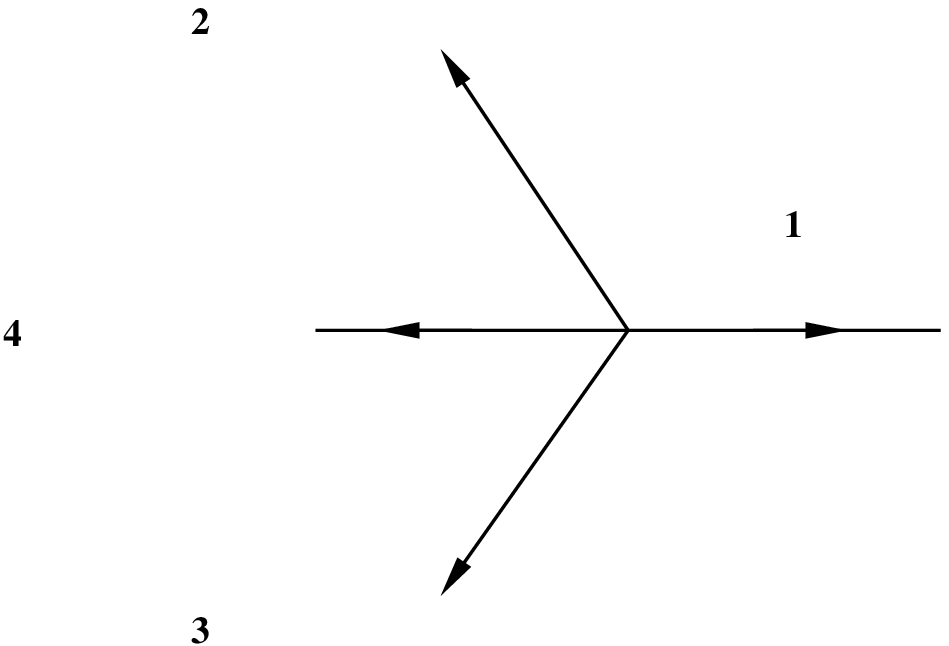}
  \end{center}
  \caption{The model RHP near $z=\beta$ which can be solved explicitly using Airy functions}
  \label{fig5}
\end{figure}
\begin{itemize}
	\item $A^{RH}(\zeta)$ is analytic for $\zeta\in\mathbb{C}\backslash\{\textnormal{arg}\ \zeta=-\frac{2\pi}{3},0,\frac{2\pi}{3},\pi\}$
	\item The following jumps hold, with the contours oriented as shown in Figure \ref{fig5}
	\begin{eqnarray*}
		A^{RH}_+(\zeta)&=&A^{RH}_-(\zeta)\begin{pmatrix}
		1 & 0\\
		-1 & 1\\
		\end{pmatrix},\hspace{0.5cm}\textnormal{arg}\ \zeta=\mp \frac{2\pi}{3}\\
		A^{RH}_+(\zeta)&=&A^{RH}_-(\zeta)\begin{pmatrix}
		1 & 1\\
		0 & 1\\
		\end{pmatrix},\hspace{0.5cm}\textnormal{arg}\ \zeta=0\\
		A^{RH}_+(\zeta)&=&A^{RH}_-(\zeta)\begin{pmatrix}
		0 & -1\\
		1 & 0\\
		\end{pmatrix},\hspace{0.5cm}\textnormal{arg}\ \zeta=\pi
	\end{eqnarray*}
	\item In order to determine the behavior of $A^{RH}(\zeta)$ at infinity we make the following observations. First let $\textnormal{arg}\ \zeta\in(-\pi,-\frac{2\pi}{3})$ and consider
	\begin{equation*}
		e^{-\frac{2}{3}\zeta^{3/2}\sigma_3}\begin{pmatrix}
		1 & 0\\
		1 & 1\\
		\end{pmatrix}e^{\frac{2}{3}\zeta^{3/2}\sigma_3}=\begin{pmatrix}
		1 & 0\\
		e^{\frac{4}{3}\zeta^{3/2}} & 1\\
		\end{pmatrix}
	\end{equation*}
	however here $\textnormal{Re}(\zeta^{3/2})<0$, hence the given product approaches the identity exponentially fast as $\zeta\rightarrow\infty$. Secondly for $\textnormal{arg}\ \zeta\in(\frac{2\pi}{3},\pi):$
	\begin{equation*}
		e^{-\frac{2}{3}\zeta^{3/2}\sigma_3}\begin{pmatrix}
		1 & 0\\
		-1 & 1\\
		\end{pmatrix}e^{\frac{2}{3}\zeta^{3/2}\sigma_3}=\begin{pmatrix}
		1 & 0\\
		-e^{\frac{4}{3}\zeta^{3/2}} & 1\\
		\end{pmatrix}
	\end{equation*}
	and also in this situation $\textnormal{Re}(\zeta^{3/2})<0$. Both cases together with the previously stated asymptotics for $A_0(\zeta)$ imply therefore
	\begin{eqnarray}\label{ARHasyinfinity}
		A^{RH}(\zeta)&=&\frac{\zeta^{\sigma_3/4}}{2\sqrt{\pi}}\begin{pmatrix}
		-1 & i\\
		1 & i\\
		\end{pmatrix}\Bigg[I+\frac{1}{48\zeta^{3/2}}\begin{pmatrix}
	1 & 6i\\
	6i & -1\\
	\end{pmatrix}+\frac{35}{4608\zeta^{6/2}}\begin{pmatrix}
	-1 & 12i\\
	-12i & -1\\
	\end{pmatrix}\nonumber\\
	&&+O\big(\zeta^{-9/2}\big)\Bigg]e^{-\frac{2}{3}\zeta^{3/2}\sigma_3}.
\end{eqnarray}
	as $\zeta\rightarrow\infty$ in a full neighborhood of infinity. 
\end{itemize}
The model function $A^{RH}(\zeta)$ will be useful in the construction of the parametrix to the solution of the $S$-RHP in a neighborhood of $z=\beta$. We proceed in two steps. First define
\begin{equation}\label{changeright}
	\zeta(z)=\bigg(\frac{3N}{4}\bigg)^{2/3}\Big(-2g(z)+V(z)+l\Big)^{2/3},\hspace{0.5cm} |z-\beta|<r.
\end{equation}
This change of variables is locally conformal, since
\begin{equation*}
	\zeta(z)=\bigg(\frac{2N}{\beta\sqrt{\beta-\alpha}}\bigg)^{2/3}(z-\beta)\big(1+O(z-\beta)\big),\ \ \ |z-\beta|<r,
\end{equation*}
and it enables us to define the right parametrix $U(z)$ near $z=\beta$ by the formula:
\begin{equation}\label{pararight}
	U(z) = B_r(z)(-i\sqrt{\pi})A^{RH}\big(\zeta(z)\big)e^{\frac{2}{3}\zeta^{3/2}(z)\sigma_3},\hspace{0.5cm}|z-\beta|<r
\end{equation}
with $\zeta(z)$ as in \eqref{changeright} and the matrix multiplier
\begin{equation}\label{multiplierright}
	B_r(z) = \begin{pmatrix}
	-i & i\\
	1 & 1\\
	\end{pmatrix}\bigg(\zeta(z)\frac{z-\alpha}{z-\beta}\bigg)^{-\sigma_3/4},\ \  B_r(\beta)=\begin{pmatrix}
	-i & i\\
	1 & 1\\
	\end{pmatrix}\bigg(\frac{2N}{\beta}(\beta-\alpha)\bigg)^{-\sigma_3/6}.
\end{equation}
By construction, in particular since $B_r(z)$ is analytic in a neighborhood of $z=\beta$, the parametrix $U(z)$ has jumps along the curves depicted in Figure \ref{fig6}, and we can always locally match the latter curves with the jump curves of the original RHP.
\begin{figure}[tbh]
  \begin{center}
  \psfragscanon
  \psfrag{2}{\footnotesize{$\begin{pmatrix}
  1 & 0\\
  e^{-NG(z)} & 1\\
  \end{pmatrix}$}}
  \psfrag{4}{\footnotesize{$\begin{pmatrix}
  1 & 0\\
  e^{NG(z)} & 1\\
  \end{pmatrix}$}}
  \psfrag{1}{\footnotesize{$\begin{pmatrix}
  1 & e^{N(g_++g_--V-l)}\\
  0 & 1\\
  \end{pmatrix}$}}
  \psfrag{3}{\footnotesize{$\begin{pmatrix}
  0 & 1\\
  -1 & 0\\
  \end{pmatrix}$}}
  \includegraphics[width=6cm,height=3cm]{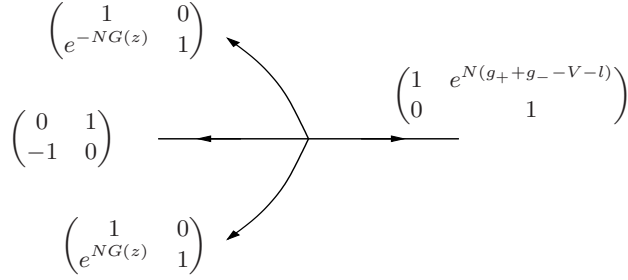}
  \end{center}
  \caption{Transformation of parametrix jumps to original jumps}
  \label{fig6}
\end{figure}

Also these jumps are described by the same matrices as in the original RHP, indeed
\begin{eqnarray*}
	e^{-\frac{2}{3}\zeta^{3/2}(z)\sigma_3}\begin{pmatrix}
	1 & 0\\
	1 & 1\\
	\end{pmatrix}e^{\frac{2}{3}\zeta^{3/2}(z)\sigma_3} &=& \begin{pmatrix}
	1 & 0\\
	e^{NG(z)} & 1\\
	\end{pmatrix},\ \ z\in\mathcal{U}\cap \gamma_4^-\\
	e^{-\frac{2}{3}\zeta^{3/2}(z)\sigma_3}\begin{pmatrix}
	1 & 1\\
	0 & 1\\
	\end{pmatrix}e^{\frac{2}{3}\zeta^{3/2}(z)\sigma_3} &=& \begin{pmatrix}
	1 & e^{N(g_++g_--V-l)}\\
	0 & 1\\
	\end{pmatrix},\ \ z\in\mathcal{U}\cap(\beta,\infty)\\
	e^{-\frac{2}{3}\zeta^{3/2}(z)\sigma_3}\begin{pmatrix}
	1 & 0\\
	1 & 1\\
	\end{pmatrix}e^{\frac{2}{3}\zeta^{3/2}(z)\sigma_3} &=& \begin{pmatrix}
	1 & 0\\
	e^{-NG(z)} & 1\\
	\end{pmatrix},\ \ z\in\mathcal{U}\cap \gamma_1^+\\
	e^{-\frac{2}{3}\zeta^{3/2}(z)\sigma_3}\begin{pmatrix}
	0 & 1\\
	-1 & 0\\
	\end{pmatrix}e^{\frac{2}{3}\zeta^{3/2}(z)\sigma_3} &=& \begin{pmatrix}
	0 & 1\\
	-1 & 0\\
	\end{pmatrix},\ \ z\in\mathcal{U}\cap (\alpha,\beta),
\end{eqnarray*}
hence the ratio of $S(z)$ with $U(z)$ is locally analytic, i.e.
\begin{equation}\label{localratioright}
	S(z)=N_r(z)U(z),\ \ |z-\beta|<r<\frac{\beta}{2}.
\end{equation}
Let us explain the role of the left multiplier $B_r(z)$ in the definition \eqref{pararight}. Observe that
\begin{equation*}
	B_r(z)\Big(-\frac{i}{2}\Big)\zeta^{\sigma_3/4}(z)\begin{pmatrix}
	-1 & i\\
	1 & i\\
	\end{pmatrix} = \begin{pmatrix}
	-i & i\\
	1 & 1\\
	\end{pmatrix}\bigg(\frac{z-\alpha}{z-\beta}\bigg)^{-\sigma_3/4}\frac{i}{2}\begin{pmatrix}
	1 & -i\\
	-1 & -i\\
	\end{pmatrix} = M(z).
\end{equation*}
This relation together with the asympotics \eqref{ARHasyinfinity} implies that,
\begin{eqnarray}\label{modelmatchright}
	U(z)&=&M(z)\bigg[I+\frac{1}{48\zeta^{3/2}}\begin{pmatrix}
	1 & 6i\\
	6i & -1\\
	\end{pmatrix}+\frac{35}{4608\zeta^{6/2}}\begin{pmatrix}
	-1 & 12i\\
	-12i & -1\\
	\end{pmatrix}+O\big(\zeta^{-9/2}\big)\Bigg]\nonumber\\
	&=&\Bigg[I+\frac{1}{96\zeta^{3/2}}\begin{pmatrix}
	7\delta^{-2}-5\delta^2 & i(7\delta^{-2}+5\delta^2)\\
	i(7\delta^{-2}+5\delta^2) & -(7\delta^{-2}-5\delta^2)\\
	\end{pmatrix}+\frac{35}{4608\zeta^{6/2}}\begin{pmatrix}
	-1 & 12i\\
	-12i & -1\\
	\end{pmatrix}\nonumber\\
	&&+O\big(\zeta^{-9/2}\big)\bigg]M(z)
\end{eqnarray}
as $N\rightarrow\infty$ and $0<r_1\leq |z-\beta|\leq r_2<\frac{\beta}{2}$ (so $|\zeta|\rightarrow\infty$). Since the function $\zeta(z)$ is of order $N^{2/3}$ on the latter annulus and $\delta(z)$ is bounded, equation \eqref{modelmatchright} yields the matching relation between the model functions $U(z)$ and $M(z)$,
\begin{equation*}
	U(z) = \big(I+o(1)\big)M(z),\hspace{0.5cm}N\rightarrow\infty,\ 0<r_1\leq |z-\beta|\leq r_2<\frac{\beta}{2}
\end{equation*}
which is crucial for the succesful implementation of the nonlinear steepest descent method as we shall see later on. This is the reason for chosing the left multiplier $B_r(z)$ in \eqref{pararight} in the form \eqref{multiplierright}.


\section{Construction of a parametrix at the edge point $z=\alpha$}

This time the construction is similar to the one given in the last subsection. First introduce for $\zeta\in\mathbb{C}$
\begin{equation*}
	\tilde{A}_0(\zeta)=-\bigl(\begin{smallmatrix}
	0 & 1\\
	1 & 0\\
	\end{smallmatrix}\bigr)\sigma_3A_0\big(e^{-i\pi}\zeta\big)\sigma_3
\end{equation*}
and obtain for $0<\textnormal{arg}\ \zeta<\pi$ as $\zeta\rightarrow\infty$
\begin{eqnarray*}
	\tilde{A}_0(\zeta) &=& \frac{\big(e^{-i\pi}\zeta\big)^{-\sigma_3/4}}{2\sqrt{\pi}}\begin{pmatrix}
	1 & -i\\
	1 & i\\
	\end{pmatrix}\Bigg[I+\frac{i}{48\zeta^{3/2}}\begin{pmatrix}
	-1 & 6i\\
	6i & 1\\
	\end{pmatrix}+\frac{35}{4608\zeta^{6/2}}\begin{pmatrix}
	1 & 12i\\
	-12i&1\\
	\end{pmatrix}\\
	&&+O\big(\zeta^{-9/2}\big)\Bigg]
	 e^{-\frac{2}{3}i\zeta^{3/2}\sigma_3}\begin{pmatrix}
	1 & -1\\
	0 & 1\\
	\end{pmatrix}
\end{eqnarray*}
as well as for $\pi<\textnormal{arg}\ \zeta<2\pi$
\begin{eqnarray*}
	\tilde{A}_0(\zeta) &=& \frac{\big(e^{-i\pi}\zeta\big)^{-\sigma_3/4}}{2\sqrt{\pi}}\begin{pmatrix}
	1 & -i\\
	1 & i\\
	\end{pmatrix}\Bigg[I+\frac{i}{48\zeta^{3/2}}\begin{pmatrix}
	-1 & 6i\\
	6i & 1\\
	\end{pmatrix}+\frac{35}{4608\zeta^{6/2}}\begin{pmatrix}
	1 & 12i\\
	-12i&1\\
	\end{pmatrix}\\
	&&+O\big(\zeta^{-9/2}\big)\Bigg]
	 e^{-\frac{2}{3}i\zeta^{3/2}\sigma_3},\ \ \zeta\rightarrow\infty.
\end{eqnarray*}
Next, instead of \eqref{ARH}, define
\begin{equation}\label{tildeARH}
	\tilde{A}^{RH}(\zeta) = \left\{
                                 \begin{array}{ll}
                                   \tilde{A}_0(\zeta)\begin{pmatrix}
                                   0 & 1\\
                                   -1 & 1\\
                                   \end{pmatrix}, & \hbox{arg $\zeta\in(0,\frac{\pi}{3})$,} \smallskip\\
                                   \tilde{A}_0(\zeta)\begin{pmatrix}
                          1 & 1 \\
                          0 & 1 \\
                        \end{pmatrix}, & \hbox{arg $\zeta\in(\frac{\pi}{3},\pi)$,} \smallskip \\
                                   \tilde{A}_0(\zeta),& \hbox{arg $\zeta\in(\pi,\frac{5\pi}{3})$,} \smallskip \\
                                   \tilde{A}_0(\zeta)\begin{pmatrix}
                                   1 & 0 \\
                                   1 & 1\\
                                   \end{pmatrix}, & \hbox{arg $\zeta\in(\frac{5\pi}{3},2\pi)$.}
                                 \end{array}
                               \right.
\end{equation}
which solves the RHP of Figure \ref{fig7}
\begin{figure}[tbh]
  \begin{center}
  \psfragscanon
  \psfrag{1}{\footnotesize{$\begin{pmatrix}
  0 & 1\\
  -1 & 0\\
  \end{pmatrix}$}}
  \psfrag{2}{\footnotesize{$\begin{pmatrix}
  1 & 0\\
  1 & 1\\
  \end{pmatrix}$}}
  \psfrag{3}{\footnotesize{$\begin{pmatrix}
  1 & 0\\
  1 & 1\\
  \end{pmatrix}$}}
  \psfrag{4}{\footnotesize{$\begin{pmatrix}
  1 & 1\\
  0 & 1\\
  \end{pmatrix}$}}
  \includegraphics[width=5cm,height=4cm]{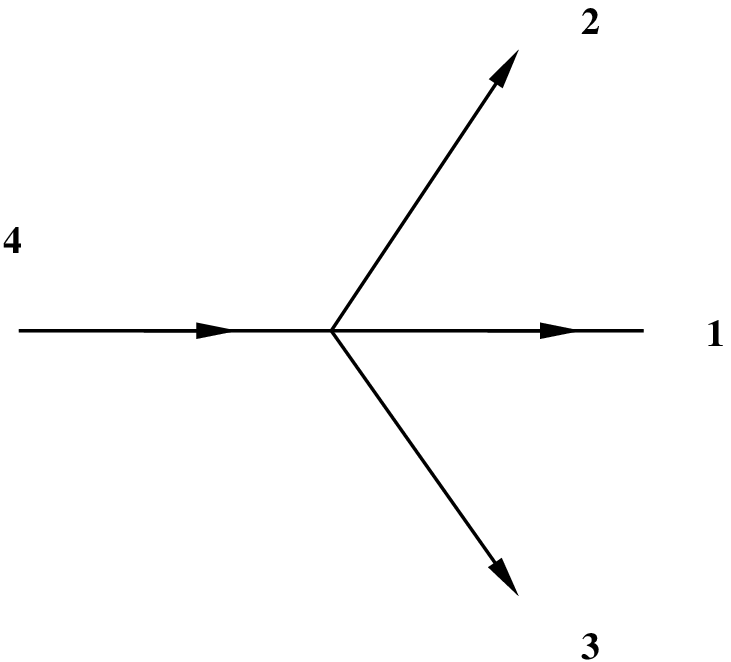}
  \end{center}
  \caption{The model RHP near $z=\alpha$ which can be solved explicitly using Airy functions}
  \label{fig7}
\end{figure}
\begin{itemize}
	\item $\tilde{A}^{RH}(\zeta)$ is analytic for $\zeta\in\mathbb{C}\backslash\{\textnormal{arg}\ \zeta=0,\frac{\pi}{3},\pi,\frac{5\pi}{3}\}$
	\item We have the following jumps on the contour depicted in Figure \ref{fig7}
	\begin{eqnarray*}
		\tilde{A}^{RH}_+(\zeta)&=&\tilde{A}^{RH}_-(\zeta)\begin{pmatrix}
		0 & 1\\
		-1 & 0\\
		\end{pmatrix},\hspace{0.5cm}\textnormal{arg}\ \zeta=0\\
		\tilde{A}^{RH}_+(\zeta)&=&\tilde{A}^{RH}_-(\zeta)\begin{pmatrix}
		1 & 0\\
		1 & 1\\
		\end{pmatrix},\hspace{0.5cm}\textnormal{arg}\ \zeta=\frac{\pi}{3},\frac{5\pi}{3}\\
		\tilde{A}^{RH}_+(\zeta)&=&\tilde{A}^{RH}_-(\zeta)\begin{pmatrix}
		1 & 1\\
		0 & 1\\
		\end{pmatrix},\hspace{0.5cm}\textnormal{arg}\ \zeta=\pi
	\end{eqnarray*}
	\item A similar argument as given in the construction of $A^{RH}(\zeta)$ implies
	\begin{eqnarray}\label{tildeARHasyinfinity}
		\tilde{A}_{RH}(\zeta)&=&\frac{\big(e^{-i\pi}\zeta\big)^{-\sigma_3/4}}{2\sqrt{\pi}}\begin{pmatrix}
	1 & -i\\
	1 & i\\
	\end{pmatrix}\Bigg[I+\frac{i}{48\zeta^{3/2}}\begin{pmatrix}
	-1 & 6i\\
	6i & 1\\
	\end{pmatrix}\nonumber\\
	&&+\frac{35}{4608\zeta^{6/2}}\begin{pmatrix}
	1 & 12i\\
	-12i&1\\
	\end{pmatrix}
	+O\big(\zeta^{-9/2}\big)\Bigg]
	 e^{-\frac{2}{3}i\zeta^{3/2}\sigma_3},\ \ \zeta\rightarrow\infty
	\end{eqnarray}
	valid in a full neighborhood of infinity.
\end{itemize}
Again we use the model function $\tilde{A}^{RH}(\zeta)$ in the construction of the parametrix to the solution of the $S$-RHP near $z=\alpha$. Instead of \eqref{changeright}
\begin{equation}\label{changeleft}
	\zeta(z)=e^{i\pi}\bigg(\frac{3N}{4}\bigg)^{2/3}\Big(-2g(z)+V(z)+l+2\pi i\ \textnormal{sgn}(\textnormal{Im}z)\Big)^{2/3},\hspace{0.5cm} |z-\alpha|<r.
\end{equation}
This change of the independent variable is locally conformal
\begin{equation*}
	\zeta(z) = \bigg(\frac{2N}{(-\alpha)\sqrt{\beta-\alpha}}\bigg)^{2/3}(z-\alpha)\big(1+O(z-\alpha)\big),\hspace{0.5cm}|z-\alpha|<r
\end{equation*}
and allows us to define the left parametrix $X^l(z)$ near $z=\alpha$ by the formula:
\begin{equation}\label{paraleft}
	V(z) = B_l(z)i\sqrt{\pi}\sigma_3\tilde{A}^{RH}\big(\zeta(z)\big)e^{\frac{2}{3}i\zeta^{3/2}(z)\sigma_3},\hspace{0.5cm}|z-\alpha|<r
\end{equation}
with the matrix multiplier
\begin{equation}\label{multiplierleft}
	B_l(z)=\begin{pmatrix}
	-i & i\\
	1 & 1\\
	\end{pmatrix}\bigg(e^{-i\pi}\zeta(z)\frac{z-\beta}{z-\alpha}\bigg)^{\sigma_3/4},\ B_l(\alpha)=\begin{pmatrix}
	-i & i\\
	1 & 1\\
	\end{pmatrix}\bigg(\frac{2N}{-\alpha}(\beta-\alpha)\bigg)^{\sigma_3/6}.
\end{equation}
Similar to the situation in the last subsection, $V(z)$ has jumps on the contour depicted in Figure \ref{fig8} which are described by exactly the same jump matrices as in the $S$-RHP, hence the ratio of parametrix $V(z)$ with $S(z)$ is locally analytic
\begin{equation}\label{localratioleft}
	S(z)= N_l(z)V(z),\ \ |z-\alpha|<r<\frac{|\alpha|}{2}
\end{equation}
\begin{figure}[tbh]
  \begin{center}
  \psfragscanon
  \psfrag{2}{\footnotesize{$\begin{pmatrix}
  1 & 0\\
  e^{-NG(z)} & 1\\
  \end{pmatrix}$}}
  \psfrag{3}{\footnotesize{$\begin{pmatrix}
  1 & 0\\
  e^{NG(z)} & 1\\
  \end{pmatrix}$}}
  \psfrag{4}{\footnotesize{$\begin{pmatrix}
  1 & e^{N(g_++g_--V-l)}\\
  0 & 1\\
  \end{pmatrix}$}}
  \psfrag{1}{\footnotesize{$\begin{pmatrix}
  0 & 1\\
  -1 & 0\\
  \end{pmatrix}$}}
  \includegraphics[width=6cm,height=3cm]{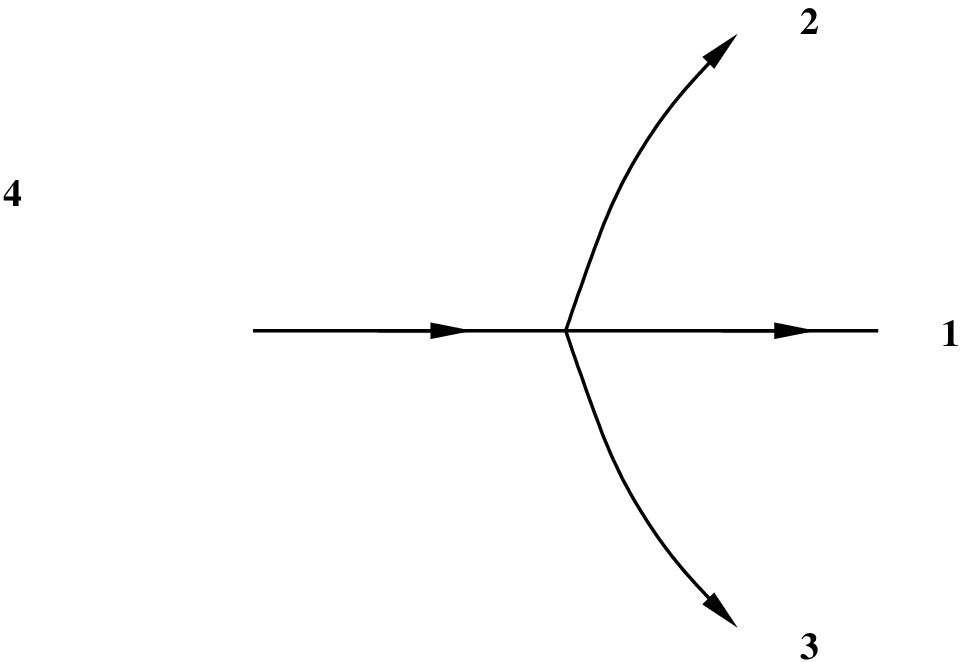}
  \end{center}
  \caption{Transformation of parametrix jumps to original jumps}
  \label{fig8}
\end{figure}
and the left multiplier \eqref{multiplierleft} in \eqref{paraleft} provides us with the following asymptotic matchup between $V(z)$ and $M(z)$
\begin{eqnarray}\label{modelmatchleft}
	V(z) &=& M(z)\bigg[I+\frac{i}{48\zeta^{3/2}}\begin{pmatrix}
	-1 & 6i\\
	6i & 1\\
	\end{pmatrix}+\frac{35}{4608\zeta^{6/2}}\begin{pmatrix}
	1 & 12i\\
	-12i & 1\\
	\end{pmatrix}+O\big(\zeta^{-9/2}\big)\bigg]\nonumber\\
	&=&\Bigg[I+\frac{i}{96\zeta^{3/2}}\begin{pmatrix}
	5\delta^{-2}-7\delta^2 & i(5\delta^{-2}+7\delta^2)\\
	i(5\delta^{-2}+7\delta^2) & -(5\delta^{-2}-7\delta^2)\\
	\end{pmatrix}+\frac{35}{4608\zeta^{6/2}}\begin{pmatrix}
	1 & 12i\\
	-12i & 1\\
	\end{pmatrix}\nonumber\\
	&&+O\big(\zeta^{-9/2}\big)\Bigg]M(z)
\end{eqnarray}
as $N\rightarrow\infty$ and $0<r_1\leq |z-\alpha|\leq r_2<\frac{|\alpha|}{2}$, thus
\begin{equation*}
	V(z) = \big(I+o(1)\big)M(z),\hspace{0.5cm}N\rightarrow\infty,\ \ 0<r_1\leq |z-\alpha|\leq r_2<\frac{|\alpha|}{2}.
\end{equation*}
At this point we can use the model functions $M(z),U(z)$ and $V(z)$ to employ the final transformation.


\section{Third and final transformation of the RHP}

In this final transformation we put
\begin{equation}\label{ratioRHP}
	R(z)=S(z)\left\{
                                   \begin{array}{ll}
                                     \big(V(z)\big)^{-1}, & \hbox{$|z-\alpha|<r$,} \\
                                     \big(U(z)\big)^{-1}, & \hbox{$|z-\beta|<r$,} \\
                                     \big(M(z)\big)^{-1}, & \hbox{$|z-\alpha|>r,\ |z-\beta|>r$}
                                   \end{array}
                                 \right.
\end{equation}
where $0<r<\min\big\{\frac{1}{2},\frac{|\alpha|}{2},\frac{\beta}{2}\big\}$ is fixed. With $C_{\alpha,\beta}$ denoting the clockwise oriented circles shown in Figure \ref{fig9}, the ratio-function $R(z)$ solves the following RHP 
\begin{figure}[tbh]
  \begin{center}
  \psfragscanon
  \psfrag{1}{\footnotesize{$C_{\alpha}$}}
  \psfrag{2}{\footnotesize{$C_{\beta}$}}
  \psfrag{3}{\footnotesize{$\hat{\gamma}_1^+$}}
  \psfrag{4}{\footnotesize{$\hat{\gamma}_2^+$}}
  \psfrag{5}{\footnotesize{$\hat{\gamma}_3^-$}}
  \psfrag{6}{\footnotesize{$\hat{\gamma}_4^-$}}
  \includegraphics[width=9cm,height=3cm]{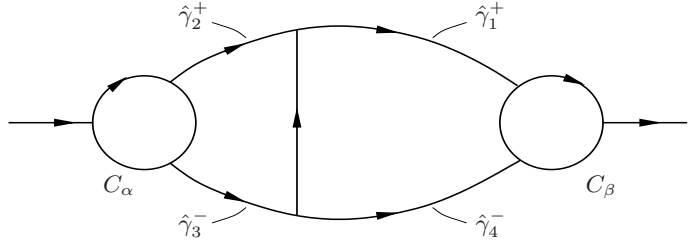}
  \end{center}
  \caption{The jump graph for the ratio function $R(z)$}
  \label{fig9}
\end{figure}
\begin{itemize}
	\item $R(z)$ is analytic for $z\in\mathbb{C}\backslash\big\{ C_{\alpha,\beta}\cup\hat{\Gamma}\cup(-\infty,\alpha-r)\cup(\beta+r,\infty)\big\}$ with $\hat{\Gamma}=\hat{\gamma}_1^+\cup\hat{\gamma}_2^+\cup\hat{\gamma}_3^-\cup\hat{\gamma}_4^-$
	\item For the jumps, along the infinite branches $(-\infty,\alpha-r)\cup(\beta+r,\infty)$
	\begin{equation*}
		R_+(z)=R_-(z)M(z)\begin{pmatrix}
		1 & e^{N(g_++g_--V-l)}\\
		0 & 1\\
		\end{pmatrix}\big(M(z)\big)^{-1},
	\end{equation*}
	on the vertical line segment $(-i\varepsilon,i\varepsilon)$
	\begin{equation*}
		R_+(z) = R_-(z)M(z)\begin{pmatrix}
		1 & 0\\
		j_k(z)& 1\\
		\end{pmatrix}\big(M(z)\big)^{-1},\ \ k=1,2,
	\end{equation*}
	on the upper lense boundary $\hat{\gamma}_1^+\cup\hat{\gamma}_2^+$ respectively lower lense boundary $\hat{\gamma}_3^-\cup\hat{\gamma}_4^-$
	\begin{eqnarray*}
		R_+(z) &=& R_-(z)M(z)\begin{pmatrix}
		1 & 0\\
		e^{-NG(z)} & 1\\
		\end{pmatrix}\big(M(z)\big)^{-1},\ \ \ z\in\hat{\gamma}_1^+\cup\hat{\gamma}_2^+\\
	 R_+(z) &=& R_-(z)M(z)\begin{pmatrix}
		1 & 0\\
		e^{NG(z)}& 1\\
		\end{pmatrix}\big(M(z)\big)^{-1},\ \ \ z\in\hat{\gamma}_3^-\cup\hat{\gamma}_4^-
	\end{eqnarray*}
	and on the clockwise oriented circles $C_{\alpha,\beta}$
	\begin{equation*}
		R_+(z) = R_-(z)\left\{
                                   \begin{array}{ll}
                                     V(z)\big(M(z)\big)^{-1}, & \hbox{$|z-\alpha|=r$,} \\
                                     U(z)\big(M(z)\big)^{-1}, & \hbox{$|z-\beta|=r$.} 
                                   \end{array}
                                 \right.
	\end{equation*}
	\item In a neigborhood of infinity, we have $R(z)\rightarrow I$.
\end{itemize}
Here, by construction, the function $R(z)$ has no jumps inside of $C_{\alpha}$ and $C_{\beta}$ and across the line segment in between. It is important to recall the previously stated behavior of the jump matrices as $N\rightarrow\infty$. In fact, on the lense boundaries, the vertical line segment $(-i\varepsilon,i\varepsilon)$ and the half rays $(-\infty,\alpha-r)\cup(\beta+r,\infty)$ the stated jump matrices approach the identity matrix. Also in virtue of \eqref{modelmatchright} and \eqref{modelmatchleft} the same holds true on the circles $C_{\alpha,\beta}$, together, with $G_R$ denoting the jump matrix in the latter ratio-RHP and $\Sigma_R$ the underlying contour,
\begin{equation}\label{DZesti1}
	\|G_R-I\|_{L^2\cap L^{\infty}(\Sigma_R)}\leq \frac{c}{\ln N},\ \ \ N\rightarrow\infty
\end{equation}
with a constant $c>0$ whose value is not important. The latter estimation enables us to solve the ratio-RHP iteratively.


\section{Solution of the RHP for $R(z)$ via iteration}

The stated RHP for the function $R(z)$
\begin{itemize}
	\item $R(z)$ is analytic for $z\in\mathbb{C}\backslash\Sigma_R$.
	\item Along the contour depicted in Figure \ref{fig9}
	\begin{equation*}
		R_+(z)=R_-(z)G_R(z),\ \ \ z\in\Sigma_R.
	\end{equation*}
	\item As $z\rightarrow\infty$, we have $R(z)=I+O\big(z^{-1}\big)$.
\end{itemize}
is equivalent to the singular integral equation 
\begin{equation}\label{integraleq}
	R_-(z)=I+\frac{1}{2\pi i}\int\limits_{\Sigma_R}R_-(w)\big(G_R(w)-I\big)\frac{dw}{w-z_-}
\end{equation}
and by standard arguments (see \cite{DZ}) we know that for sufficiently large $N$ the relevant integral operator is contracting and equation \eqref{integraleq} can be solved iteratively in $L^2(\Sigma_R)$. Moreover, its unique solution satisfies
\begin{equation}\label{DZesti2}
	\|R_--I\|_{L^2(\Sigma_R)}\leq \frac{c}{\ln N},\ \ \ N\rightarrow\infty.
\end{equation}
Observe that for $z\in\mathbb{C}\backslash\Sigma_R$
\begin{equation}\label{integralasy1}
	R(z)=I+\frac{i}{2\pi z}\int\limits_{\Sigma_R}R_-(w)\big(G_R(w)-I\big)dw+O\big(z^{-2}\big),\ \ z\rightarrow\infty
\end{equation}
and also as $N\rightarrow\infty$ following from \eqref{DZesti2}, \eqref{modelmatchright}, \eqref{modelmatchleft} and \eqref{DZesti1} as well as the previous discussion about exponentially small contributions
\begin{equation}\label{integralasy2}
	\int\limits_{\Sigma_R}R_-(w)\big(G_R(w)-I\big)dw = \int\limits_{C_{\alpha,\beta}}\big(G_R(w)-I\big)dw+\int\limits_{-i\varepsilon}^{i\varepsilon}\big(G_R(w)-I\big)dw+O\big(N^{-2}\big).
\end{equation}
We are now ready to prove the statement on the large $N$ asymptotics of $h_N$ given in Theorem $1$.


\section{Asymptotics of $h_N$ - proof of theorem $1$}

We recall the following identities, see \eqref{constantRHP},
\begin{equation*}
	h_{N,n}=-2\pi i\big(Y_1^{(n)}\big)_{12},\ \ \ \ h_n=N^{2n+1}h_{N,n}
\end{equation*}
with
\begin{equation*}
	Y^{(n)}(z) = \bigg(I+\frac{Y_1^{(n)}}{z}+O\big(z^{-2}\big)\bigg)z^{n\sigma_3},\ \ z\rightarrow\infty.
\end{equation*}
and trace back the transformations
\begin{equation*}
	Y(z)\equiv Y^{(N)}(z) \mapsto T(z) \mapsto S(z) \mapsto R(z).
\end{equation*}
Thus
\begin{equation*}
	Y_1^{(N)} = \lim_{z\rightarrow\infty}\Big(z\big(Y^{(N)}(z)z^{-N\sigma_3}-I\big)\Big) = \lim_{z\rightarrow\infty}\Big(z\big(e^{\frac{Nl}{2}\sigma_3}R(z)M(z)e^{N(g(z)-\frac{l}{2})\sigma_3}z^{-N\sigma_3}-I\big)\Big).
\end{equation*}
and since
\begin{equation*}
	e^{N(g(z)-\frac{l}{2})\sigma_3}z^{-N\sigma_3} = e^{-\frac{Nl}{2}\sigma_3}\Big(I-\frac{N(\alpha+\beta)}{4z}\sigma_3+O\big(z^{-2}\big)\Big),\ \ z\rightarrow\infty
\end{equation*}
and similarly
\begin{equation*}
	M(z)=\begin{pmatrix}
	-i & i\\
	1 & 1\\
	\end{pmatrix}\delta(z)^{-\sigma_3}\frac{i}{2}\begin{pmatrix}
	1 & -i\\
	-1 & -i\\
	\end{pmatrix} = I-\frac{\beta-\alpha}{4z}\sigma_2+O\big(z^{-2}\big),\ \ z\rightarrow\infty
\end{equation*}
we have in combination with \eqref{integralasy1} and \eqref{integralasy2} as $N\rightarrow\infty$
\begin{eqnarray}\label{exactY1}
	&&e^{-\frac{Nl}{2}\sigma_3}Y_1^{(N)}e^{\frac{Nl}{2}\sigma_3} = -N(\beta+\alpha)\frac{\sigma_3}{4}-(\beta-\alpha)\frac{\sigma_2}{4} +\frac{i}{2\pi}\int\limits_{\Sigma_R}R_-(w)\big(G_R(w)-I\big)dw\nonumber\\
	&=&-N(\beta+\alpha)\frac{\sigma_3}{4}-(\beta-\alpha)\frac{\sigma_2}{4} +\frac{i}{2\pi}\int\limits_{C_{\alpha,\beta}}\big(G_R(w)-I\big)dw+\frac{i}{2\pi}\int\limits_{-i\varepsilon}^{i\varepsilon}\big(G_R(w)-I\big)dw\nonumber\\
	&&+O\big(N^{-2}\big)
\end{eqnarray}
We start computing the contributions arising from the circles $C_{\alpha,\beta}$. From 
\eqref{modelmatchright} as $N\rightarrow\infty$
\begin{equation*}
	\int\limits_{C_{\beta}}\big(G_R(w)-I\big)dw=\frac{1}{96}\int\limits_{C_{\beta}}\begin{pmatrix}
	7\delta^{-2}-5\delta^2 & i(7\delta^{-2}+5\delta^2)\\
	i(7\delta^{-2}+5\delta^2) & -(7\delta^{-2}-5\delta^2)\\
	\end{pmatrix}\frac{dw}{\zeta^{3/2}(w)}+O\big(N^{-2}\big),
\end{equation*}
and since
\begin{equation*}
	\zeta(z)=\bigg(\frac{2N}{\beta\sqrt{\beta-\alpha}}\bigg)^{2/3}(z-\beta)\bigg\{1+\frac{5\alpha\beta-3\beta^2-2\alpha^2}{10\beta(\beta-\alpha)^2}(z-\beta)+O\big((z-\beta)^2\big)\bigg\}^{2/3}
\end{equation*}
for $z:|z-\beta|<r$, we obtain via residue theorem
\begin{equation*}
	\int\limits_{C_{\beta}}\frac{\delta^{-2}(w)}{\zeta^{3/2}(w)}dw = -2\pi i\frac{\beta}{2N},\ \ \ \int\limits_{C_{\beta}}\frac{\delta^2(w)}{\zeta^{3/2}(w)}dw = -2\pi i\frac{\beta}{4N}\bigg(1-\frac{5\alpha\beta-3\beta^2-2\alpha^2}{5\beta(\beta-\alpha)}\bigg).
\end{equation*}
Combined together as $N\rightarrow\infty$
\begin{eqnarray}\label{intCr}
	&&\int\limits_{C_{\beta}}\big(G_R(w)-I\big)dw\\
	 &=&-\frac{2\pi i}{192N(\beta-\alpha)}\begin{pmatrix}
	3\beta^2-2\alpha\beta-\alpha^2 & i(11\beta^2-12\alpha\beta+\alpha^2)\\
	i(11\beta^2-12\alpha\beta+\alpha^2) & -(3\beta^2-2\alpha\beta-\alpha^2)\\
	\end{pmatrix}+O\big(N^{-2}\big).\nonumber
\end{eqnarray}
For the integral over $C_{\alpha}$ we follow the same strategy. First from \eqref{modelmatchleft} as $N\rightarrow\infty$
\begin{equation*}
	\int\limits_{C_{\alpha}}\big(G_R(w)-I\big)dw = \frac{i}{96}\int\limits_{C_{\alpha}}\begin{pmatrix}
	5\delta^{-2}-7\delta^2 & i(5\delta^{-2}+7\delta^2)\\
	i(5\delta^{-2}+7\delta^2) & -(5\delta^{-2}-7\delta^2)\\
	\end{pmatrix}\frac{dw}{\zeta^{3/2}(w)} +O\big(N^{-2}\big),
\end{equation*}
and with
\begin{equation*}
	\zeta(z) = e^{i\pi}\bigg(\frac{2N}{-\alpha\sqrt{\beta-\alpha}}\bigg)^{2/3}(\alpha-z)\bigg\{1+\frac{5\alpha\beta-3\alpha^2-2\beta^2}{10(-\alpha)(\beta-\alpha)^2}(\alpha-z)+O\big((z-\alpha)^2\big)\bigg\}^{2/3}
\end{equation*}
for $z:|z-\alpha|<r$, we deduce from residue theorem
\begin{equation*}
	\int\limits_{C_{\alpha}}\frac{\delta^{-2}(w)}{\zeta^{3/2}(w)}dw = 2\pi i\frac{(-\alpha)i}{4N}\bigg(1-\frac{5\alpha\beta-3\alpha^2-2\beta^2}{5(-\alpha)(\beta-\alpha)}\bigg),\ \ \int\limits_{C_{\alpha}}\frac{\delta^2(w)}{\zeta^{3/2}(w)}dw = 2\pi i\frac{(-\alpha)i}{2N}
\end{equation*}
so together
\begin{eqnarray}\label{intCl}
	&&\int\limits_{C_{\alpha}}\big(G_R(w)-I\big)dw\\
	 &=&-\frac{2\pi i}{192N(\beta-\alpha)}\begin{pmatrix}
	 -(3\alpha^2-2\alpha\beta-\beta^2) & i(11\alpha^2-12\alpha\beta+\beta^2)\\
	 i(11\alpha^2-12\alpha\beta+\beta^2) & 3\alpha^2-2\alpha\beta-\beta^2\\
	 \end{pmatrix}+O\big(N^{-2}\big).\nonumber  
\end{eqnarray}
Adding \eqref{intCr} and \eqref{intCl} we have thus
\begin{equation}\label{intAiry}
	\int\limits_{C_{\alpha,\beta}}\big(G_R(w)-I\big)dw = -\frac{2\pi i}{48N}\begin{pmatrix}
	\beta+\alpha & 3i(\beta-\alpha)\\
	3i(\beta-\alpha) & -(\beta+\alpha)\\
	\end{pmatrix} +O\big(N^{-2}\big),\ \ N\rightarrow\infty.
\end{equation}
Let us now move on to the evaluation of the integral
\begin{eqnarray}\label{verticallinesegment}
	\int\limits_{-i\varepsilon}^{i\varepsilon}\big(G_R(w)-I\big)dw&=&
	 \int\limits_0^{i\varepsilon}\frac{j_1(w)}{4}\begin{pmatrix}
	i(\delta^2-\delta^{-2}) & (\delta-\delta^{-1})^2\\
	(\delta+\delta^{-1})^2 & -i(\delta^2-\delta^{-2})\\
	\end{pmatrix}dw\nonumber\\
	&&+\int\limits_{-i\varepsilon}^0\frac{j_2(w)}{4}\begin{pmatrix}
	i(\delta^2-\delta^{-2}) & (\delta-\delta^{-1})^2\\
	(\delta+\delta^{-1})^2 & -i(\delta^2-\delta^{-2})\\
	\end{pmatrix}dw.
\end{eqnarray}
As we see this evaluation requires the asymptotical computation of the integrals
\begin{equation}\label{verticalintegrals}
	\int\limits_0^{i\varepsilon}j_1(w)f(w)dw\ \ \ \ \int\limits_{-i\varepsilon}^0j_2(w)f(w)dw
\end{equation}
for a function $f(w)$ which is analytic on $(-i\varepsilon,0)\cup(0,i\varepsilon)$. Consider the first integral, 
we have from \eqref{finalvertical1}
\begin{equation*}
	\int\limits_0^{i\varepsilon}j_1(w)f(w)dw = -2e^{-Nh_2(0)}\int\limits_0^{\varepsilon}e^{-N(-\frac{2}{\pi}y\ln y+yh_1(y)+h_2(y)-h_2(0))}\sin(Ny)f(iy)dy
\end{equation*}
where
\begin{equation*}
	h_1(y)=\frac{4}{\pi}\ln\frac{\sqrt{\beta(iy-\alpha)}+\sqrt{-\alpha(\beta-iy)}}{\sqrt{\beta-\alpha}},\ \ h_2(y)=4\ln\frac{\sqrt{iy-\alpha}+i\sqrt{\beta-iy}}{\sqrt{\beta-\alpha}}.
\end{equation*}
Let us perform the change of variables
\begin{equation}\label{change1}
	u = u(y) = -\frac{2}{\pi}y\ln y+yh_1(y)+h_2(y)-h_2(0),\ \ y\in[0,\varepsilon).
\end{equation}
We have $u(0)=0$ and would now like to express $y$ as a function of $u$. To this end introduce $v=u/y$, thus from \eqref{change1}
\begin{equation}\label{change2}
	v = -\frac{2}{\pi}\ln u+\frac{2}{\pi}\ln v+h_1\Big(\frac{u}{v}\Big)+\frac{v}{u}\bigg(h_2\Big(\frac{u}{v}\Big)-h_2(0)\bigg)
\end{equation}
and we are going to solve this equation for $v$ by iteration
\begin{equation*}
	v_0=-\frac{2}{\pi}\ln u,\hspace{0.5cm} v_{n+1} = -\frac{2}{\pi}\ln u+\frac{2}{\pi}\ln v_n+h_1\Big(\frac{u}{v_n}\Big)+\frac{v_n}{u}\bigg(h_2\Big(\frac{u}{v_n}\Big)-h_2(0)\bigg).
\end{equation*}
First
\begin{equation*}
	v_1=-\frac{2}{\pi}\ln u+\frac{2}{\pi}\ln\Big(-\frac{2}{\pi}\ln u\Big)+h_1(0)+h_2'(0)+O\bigg(\frac{u}{-\ln u}\bigg),\ \ u\rightarrow 0
\end{equation*}
and secondly
\begin{equation*}
 v_2=-\frac{2}{\pi}\ln u+\frac{2}{\pi}\ln\Big(-\frac{2}{\pi}\ln u\Big)+h_1(0)+h_2'(0)+O\bigg(\frac{\ln(-\ln u)}{-\ln u}\bigg),\ \ u\rightarrow 0.
\end{equation*}
This asymptotic behavior persists for subsequent $v_n$'s, we have a solution to equation \eqref{change2} of the form
\begin{equation*}
	v = -\frac{2}{\pi}\ln u+\frac{2}{\pi}\ln\Big(-\frac{2}{\pi}\ln u\Big)+h_1(0)+h_2'(0)+O\bigg(\frac{\ln(-\ln u)}{-\ln u}\bigg)
\end{equation*}
as $u\rightarrow 0$. This in turn implies for the solution $y=y(u)$ that as $u\to 0$, 
\begin{equation}\label{ychangeu}
	y=\frac{u}{-\frac{2}{\pi}\ln u}\bigg[1-\frac{\ln\big(-\frac{2}{\pi}\ln u\big)}{-\ln u}-\frac{h_1(0)+h_2'(0)}{-\frac{2}{\pi}\ln u}+O\Big(\bigg(\frac{\ln(-\ln u)}{-\ln u}\bigg)^2\Big)\bigg],
\end{equation}
as well as for its derivative
\begin{eqnarray}\label{ychangeuderiv}
	\frac{dy}{du} &=& \bigg(-\frac{2}{\pi}\ln y-\frac{2}{\pi}+h_1(y)+yh_1'(y)+h_2'(y)\bigg)^{-1}\\
	&=&\frac{1}{-\frac{2}{\pi}\ln u}\bigg[1-\frac{\ln\big(-\frac{2}{\pi}\ln u\big)}{-\ln u}+\frac{1-\frac{\pi}{2}\big(h_1(0)+h_2'(0)\big)}{-\ln u}+O\Big(\bigg(\frac{\ln(-\ln u)}{-\ln u}\bigg)^2\Big)\bigg].\nonumber
\end{eqnarray}
At this time we go back to the given integral 
\begin{eqnarray*}
	\int\limits_0^{i\varepsilon}j_1(w)f(w)dw &=& -2e^{-Nh_2(0)}\int\limits_0^{\varepsilon}e^{-N(-\frac{2}{\pi}y\ln y+yh_1(y)+h_2(y)-h_2(0))}\sin(Ny)f(iy)dy\\
	&=&-2e^{-Nh_2(0)}\int\limits_0^{u(\varepsilon)}e^{-Nu}\sin\big(Ny(u)\big)f\big(iy(u)\big)\frac{dy}{du} du
\end{eqnarray*}
and introduce
\begin{equation}\label{AN}
	A(N)=\frac{4}{\pi^2}\int\limits_0^{\varepsilon}e^{-N(-\frac{2}{\pi}y\ln y+yh_1(y)+h_2(y)-h_2(0))}\sin(Ny)dy.
\end{equation}
\begin{proposition}\label{prop1}
As $N\rightarrow\infty$
\begin{eqnarray*}
	A(N)&=&\frac{1}{N(\ln N)^2}\Bigg[1-\frac{2\ln\ln N}{\ln N}+\frac{3-2\gamma-2\ln\big(\frac{2}{\pi}\big)-\pi\big(h_1(0)+h_2'(0)\big)}{\ln N}\\
	&&+O\Big(\bigg(\frac{\ln\ln N}{\ln N}\bigg)^2\Big)\Bigg]
\end{eqnarray*}
\end{proposition}
\begin{proof} We use the change of variables \eqref{ychangeu}, \eqref{ychangeuderiv} as well as the substitution $s=Nu$,
\begin{eqnarray*}
	A(N) &=&\frac{1}{N}\int\limits_0^{Nu(\varepsilon)}\frac{se^{-s}}{(\ln N-\ln s)^2}\Bigg[1-\frac{2\ln\big(\frac{2}{\pi}\ln N-\frac{2}{\pi}\ln s\big)}{\ln N-\ln s}\\
	&&+\frac{1-\pi\big(h_1(0)+h_2'(0)\big)}{\ln N-\ln s}-\frac{\ln\big(\frac{2}{\pi}\ln N-\frac{2}{\pi}\ln s\big)}{(\ln N-\ln s)^2}\Big(1-\pi\big(h_1(0)+h_2'(0)\big)\Big)\\
	&&+O\Big(\bigg(\frac{\ln(\ln N-\ln s)}{\ln N-\ln s}\bigg)^2\Big)\Bigg]ds
\end{eqnarray*}
Since $u(\varepsilon)$ can be chosen arbitrarily small, we can expand the integrand in reciprocal powers of $\ln N$. This gives as $N\rightarrow\infty$
\begin{eqnarray*}
	A(N) &=&\frac{1}{N(\ln N)^2}\Bigg[I_1(N)-\frac{2\ln\big(\frac{2}{\pi}\ln N\big)}{\ln N}I_1(N)\\
	&&+\frac{1-\pi\big(h_1(0)+h_2'(0)\big)}{\ln N}I_1(N)+\frac{2I_2(N)}{\ln N}+O\Big(\bigg(\frac{\ln\ln N}{\ln N}\bigg)^2\Big)\Bigg]
\end{eqnarray*}
with
\begin{equation*}
	I_1(N)=\int\limits_0^{Nu(\varepsilon)}e^{-s}s\ ds,\ \ \ I_2(N)=\int\limits_0^{Nu(\varepsilon)}e^{-s}s\ln s\ ds.
\end{equation*}
Up to an exponentially small error, we have $I_1(N)=1$ and $I_2(N)=1-\gamma$, where $\gamma$ denotes Euler's constant, thus
\begin{eqnarray}\label{ANfinal}
	A(N)&=&\frac{1}{N(\ln N)^2}\Bigg[1-\frac{2\ln\ln N}{\ln N}+\frac{3-2\gamma -2\ln\big(\frac{2}{\pi}\big)-\pi\big(h_1(0)+h_2'(0)\big)}{\ln N}\nonumber\\
	&&+O\Big(\bigg(\frac{\ln\ln N}{\ln N}\bigg)^2\Big)\Bigg],\ \ \ N\rightarrow\infty.
\end{eqnarray}
\end{proof}
Back to the integral under consideration, we notice from \eqref{ychangeu}
\begin{equation*}
	f\big(iy\Big(\frac{s}{N}\Big)\big)-f_+(0) = O\bigg(\frac{1}{N\ln N}\bigg),\ \ N\rightarrow\infty,
\end{equation*}
hence via Proposition \ref{prop1}
\begin{equation}\label{firstverticalfinal}
	\int\limits_0^{i\varepsilon}j_1(w)f(w)dw=-\frac{\pi^2}{2} e^{-Nh_2(0)}f_+(0)A(N)+O\bigg(\frac{1}{N^2(\ln N)^3}\bigg)
\end{equation}
with
\begin{eqnarray*}
	&&h_1(0)=\frac{4}{\pi}\ln\frac{2\sqrt{(-\alpha)\beta}}{\sqrt{\beta-\alpha}},\ \ h_2(0) = 4\ln\frac{\sqrt{-\alpha}+i\sqrt{\beta}}{\sqrt{\beta-\alpha}}=4i\ \textnormal{arg}\ \frac{\sqrt{-\alpha}+i\sqrt{\beta}}{\sqrt{\beta-\alpha}},\\
	&&\ \ h_2'(0)=\frac{2}{\sqrt{(-\alpha)\beta}}.
\end{eqnarray*}
The second integral in \eqref{verticalintegrals} can be treated in a similar way. Indeed we have
\begin{equation*}
	\int\limits_{-i\varepsilon}^0j_2(w)f(w)dw = -2e^{-Nh_4(0)}\int\limits_0^{\varepsilon}e^{-N\big(-\frac{2}{\pi}y\ln y+yh_3(y)+h_4(y)-h_4(0)\big)}\sin(Ny)f(-iy)dy
\end{equation*}
with
\begin{equation*}
	h_3(y)=\frac{4}{\pi}\ln\frac{\sqrt{\beta(-iy-\alpha)}+\sqrt{-\alpha(\beta+iy)}}{\sqrt{\beta-\alpha}},\ \ h_4(y)=-4\ln\frac{\sqrt{-iy-\alpha}+i\sqrt{\beta+iy}}{\sqrt{\beta-\alpha}}
\end{equation*}
and we deduce as $N\rightarrow\infty$
\begin{eqnarray}\label{secondverticalfinal}
	\int\limits_{-i\varepsilon}^0j_2(w)f(w)dw&=&-\frac{\pi^2}{2}e^{-Nh_4(0)}f_-(0)A(N)+O\bigg(\frac{1}{N^2(\ln N)^3}\bigg),
\end{eqnarray}
with
\begin{equation*}
	h_4(0)=-4i\ \textnormal{arg}\ \frac{\sqrt{-\alpha}+i\sqrt{\beta}}{\sqrt{\beta-\alpha}}.
\end{equation*}
Adding \eqref{firstverticalfinal} and \eqref{secondverticalfinal}, we end up with
\begin{eqnarray}\label{verticalgeneral}
	&&\int\limits_0^{i\varepsilon}j_1(w)f(w)dw+\int\limits_{-i\varepsilon}^0j_2(w)f(w)dw = -\frac{\pi^2}{2}\Big(e^{-Nh_2(0)}f_+(0)+e^{Nh_2(0)}f_-(0)\Big)A(N)\nonumber\\
	&&+O\bigg(\frac{1}{N^2(\ln N)^3}\bigg),\ \ N\rightarrow\infty.
\end{eqnarray}
The latter expansion enables us now to evaluate \eqref{verticallinesegment}. Since
\begin{equation*}
	\delta_+(0)=\bigg(\frac{-\alpha}{\beta}\bigg)^{1/4}e^{-i\frac{\pi}{4}},\ \ \delta_-(0)=\bigg(\frac{-\alpha}{\beta}\bigg)^{1/4}e^{i\frac{\pi}{4}}
\end{equation*}
we obtain
\begin{eqnarray*}
	&&\int\limits_{-i\varepsilon}^{i\varepsilon}\big(G_R(w)-I\big)dw = -\frac{\pi^2}{8}A(N)\times\\
	&&\begin{pmatrix}
	-2i\big((\frac{-\alpha}{\beta})^{1/2}+(\frac{-\alpha}{\beta})^{-1/2}\big)\sin\varphi_N & -4\cos\varphi_N-2\big((\frac{-\alpha}{\beta})^{1/2}-(\frac{-\alpha}{\beta})^{-1/2}\big)\sin\varphi_N\\
	4\cos\varphi_N-2\big((\frac{-\alpha}{\beta})^{1/2}-(\frac{-\alpha}{\beta})^{-1/2}\big)\sin\varphi_N& 2i\big((\frac{-\alpha}{\beta})^{1/2}+(\frac{-\alpha}{\beta})^{-1/2}\big)\sin\varphi_N\\
	\end{pmatrix}\\
	&&+O\bigg(\frac{1}{N^2(\ln N)^3}\bigg)
\end{eqnarray*}
as $N\rightarrow\infty$ with
\begin{equation*}
	\varphi_N=4N\ \textnormal{arg}\frac{\sqrt{-\alpha}+i\sqrt{\beta}}{\sqrt{\beta-\alpha}}.
\end{equation*}
All together from \eqref{exactY1}
\begin{eqnarray*}
	&&e^{-\frac{Nl}{2}\sigma_3}Y_1^{(N)}e^{\frac{Nl}{2}\sigma_3} = 
	-N(\beta+\alpha)\frac{\sigma_3}{4}-(\beta-\alpha)\frac{\sigma_2}{4}+\frac{1}{48N}\begin{pmatrix}
	\beta+\alpha & 3i(\beta-\alpha)\\
	3i(\beta-\alpha)&-(\beta+\alpha)\\
	\end{pmatrix}\\
	&&-\frac{i\pi}{16}A(N)\times\\
	&&\begin{pmatrix}
	-2i\big((\frac{-\alpha}{\beta})^{1/2}+(\frac{-\alpha}{\beta})^{-1/2}\big)\sin\varphi_N & -4\cos\varphi_N-2\big((\frac{-\alpha}{\beta})^{1/2}-(\frac{-\alpha}{\beta})^{-1/2}\big)\sin\varphi_N\\
	4\cos\varphi_N-2\big((\frac{-\alpha}{\beta})^{1/2}-(\frac{-\alpha}{\beta})^{-1/2}\big)\sin\varphi_N& 2i\big((\frac{-\alpha}{\beta})^{1/2}+(\frac{-\alpha}{\beta})^{-1/2}\big)\sin\varphi_N\\
	\end{pmatrix}\\
	&&+O\big(N^{-2}\big),\ \ N\rightarrow\infty
\end{eqnarray*}
and in particular
\begin{eqnarray*}
	\Big(Y_1^{(N)}\Big)_{12} &=& e^{Nl}\Bigg[\frac{i(\beta-\alpha)}{4}+\frac{i(\beta-\alpha)}{16N}+\frac{i\pi}{16}A(N)\bigg(4\cos\varphi_N\\
	&&+2\Big\{\Big(\frac{-\alpha}{\beta}\Big)^{1/2}-\Big(\frac{-\alpha}{\beta}\Big)^{-1/2}\Big\}\sin\varphi_N\bigg)+O\big(N^{-2}\big)\Bigg].
\end{eqnarray*}
The latter expansion allows us to deduce the asymptotics of the normalization constants $h_N$ as $N\rightarrow\infty$. Since
\begin{equation*}
	h_N=N^{2N+1}h_{N,N} = -2\pi iN^{2N+1}\Big(Y_1^{(N)}\Big)_{12}
\end{equation*}
one obtains
\begin{eqnarray*}
	h_N&=&\frac{\pi(\beta-\alpha)}{2} N^{2N+1}e^{Nl}\Bigg[1+\frac{1}{4N}+\frac{\pi A(N)}{4(\beta-\alpha)}\\
	&&\times\bigg(4\cos\varphi_N+2\Big\{\Big(\frac{-\alpha}{\beta}\Big)^{1/2}-\Big(\frac{-\alpha}{\beta}\Big)^{-1/2}\Big\}\sin\varphi_N\bigg)+O\big(N^{-2}\big)\Bigg],\ \ N\rightarrow\infty
\end{eqnarray*}
or in other words, recalling the definitions of $\alpha,\beta$ and $l$ and the identity
\begin{equation*}
	\varphi_N=4N\textnormal{arg}\frac{\sqrt{-\alpha}+i\sqrt{\beta}}{\sqrt{\beta-\alpha}} = 4N\arctan\bigg(\frac{1}{\tan\frac{\pi}{4}(1-x)}\bigg)=\pi N(1+x),
\end{equation*}
we deduce
\begin{eqnarray}\label{hNalone}
	h_N&=&\frac{\pi}{2}\bigg(\frac{2\pi N}{\cos\frac{\pi x}{2}}\bigg)^{2N+1}\frac{e^{-2N}}{16^N}\Bigg[1+\frac{1}{4N}+\frac{(-1)^N}{2}\cos\Big(\pi x\bigg(N+\frac{1}{2}\bigg)\Big)A(N)\nonumber\\
	&&+O\big(N^{-2}\big)\Bigg]
\end{eqnarray}
Here the stated expansion as $N\rightarrow\infty$ is uniform on any compact subset of the set \eqref{excset}. Furthermore by Stirling's approximation
\begin{equation*}
	N! = \bigg(\frac{N}{e}\bigg)^N\sqrt{2\pi N}\Big(1+\frac{1}{12N}+O\big(N^{-2}\big)\Big),\ \ N\rightarrow\infty
\end{equation*}
one obtains
\begin{equation}\label{theo1rewrite}
	\frac{h_N}{(N!)^2} =\bigg(\frac{\pi}{2\cos\frac{\pi t}{2}}\bigg)^{2N+1}e^{b_N-\frac{1}{6N}}
\end{equation}
with
\begin{eqnarray*}
	b_N&=&\frac{1}{4N}+\frac{(-1)^N}{2}\cos\Big(\pi x\bigg(N+\frac{1}{2}\bigg)\Big)A(N)+O\big(N^{-2}\big)\\
	&=&\frac{1}{4N}+\frac{(-1)^N\cos(\pi x\big(N+\frac{1}{2}\big))}{2N(\ln N)^2}\bigg\{1-\frac{2\ln\ln N}{\ln N}+\frac{1-2\gamma-4\ln 2-2\ln\big(\cos\frac{\pi x}{2}\big)}{\ln N}\\
	&&+O\Big(\bigg(\frac{\ln\ln N}{\ln N}\bigg)^2\Big)\bigg\}+O\big(N^{-2}\big),\ \ N\rightarrow\infty,
\end{eqnarray*}
thus proving Theorem \ref{theo1}.

\section{Asymptotics of the partition function $Z_N$ - proof of Theorem \ref{theo2}}

We go back to \eqref{IKZformula}
\begin{equation}\label{IZconst}
	Z_N=\frac{(ab)^{N^2}}{(\prod_{k=0}^{N-1}k!)^2}\tau_N = \big(1-x^2\big)^{N^2}\prod_{k=0}^{N-1}\frac{h_k}{(k!)^2}
\end{equation}
and derive in the given situation from \eqref{theo1rewrite}
\begin{equation}\label{ZNalmost}
	Z_N=CF^{N^2}e^{\sum_{k=1}^{N-1}(b_k-\frac{1}{6k})},\ \ \ F=\frac{\pi(1-x^2)}{2\cos\frac{\pi x}{2}}
\end{equation}
with a positive, $N$ independent, constant $C$. Applying now Euler's summation formula, we have as $N\rightarrow\infty$
\begin{equation*}
	\sum_{k=1}^{N-1}\Big(b_k-\frac{1}{6k}\Big) = \frac{1}{12}\ln N+c_0+\frac{1}{2}\sum_{k=1}^{N-1}(-1)^k\cos\big(\pi x\Big(k+\frac{1}{2}\Big)\big)A(k)+O\big(N^{-1}\big),
\end{equation*}
with an $N$ independent term $c_0$. The sum can be further estimated using summation by parts: 
\begin{eqnarray}\label{byparts}
	\sum_{k=1}^{N-1}(-1)^k\cos\big(\pi x\Big(k+\frac{1}{2}\Big)\big)A(k) &=& A(N-1)S(N-1)\\
	&&-\sum_{k=1}^{N-2}\big(A(k+1)-A(k)\big)S(k)\nonumber
\end{eqnarray}
with
\begin{equation*}
	S(k) = \sum_{l=1}^k(-1)^l\cos\big(\pi x\Big(l+\frac{1}{2}\Big)\big) = \frac{(-1)^k\cos(\pi x(k+1))-\cos\pi x}{2\cos\frac{\pi}{2}x}.
\end{equation*}
Following the notation of Proposition \ref{prop1}, as $N\rightarrow\infty$
\begin{equation*}
	A(N+1) 
	=\frac{4}{\pi^2}\int\limits_0^{u(\varepsilon)}e^{-(N+1)u}\sin\big(Ny(u)\big)\frac{dy}{du}du +O\bigg(\frac{1}{N^2(\log N)^3}\bigg)
\end{equation*}
hence
\begin{eqnarray*}
	A(N+1)-A(N)&=&\frac{4}{\pi^2}\int\limits_0^{u(\varepsilon)}e^{-Nu}\Big(e^{-u}-1\Big)\sin\big(Ny(u)\big)\frac{dy}{du}du + O\bigg(\frac{1}{N^2(\log N)^3}\bigg)\\
	&=&O\bigg(\frac{1}{N^2(\log N)^2}\bigg),\ \ \ N\rightarrow\infty
\end{eqnarray*}
and therefore the series 
\begin{equation*}
	\sum_{k=1}^{\infty}\big(A(k+1)-A(k)\big)S(k)
\end{equation*}
is absolutely and uniformly convergent on any compact subset of the set \eqref{excset}. Back to \eqref{byparts} using Proposition \ref{prop1}
\begin{eqnarray*}
	&&\sum_{k=1}^{N-1}(-1)^k\cos\big(\pi x\Big(k+\frac{1}{2}\Big)\big)A(k) = A(N-1)S(N-1) +C_0\\
	&&+\sum_{k=N-1}^{\infty}\big(A(k+1)-A(k)\big)S(k) = C_0 +O\bigg(\frac{1}{N(\ln N)^2}\bigg),\ \ N\rightarrow\infty
\end{eqnarray*}
with an $N$ independent term $C_0$, thus
\begin{equation}\label{erroresti}
	\sum_{k=1}^{N-1}\Big(b_k-\frac{1}{6k}\Big) = \frac{1}{12}\ln N+\hat{C} +O\big(N^{-1}\big),\ \ N\rightarrow\infty.
\end{equation}
Back to \eqref{ZNalmost}, we get as $N\rightarrow\infty$
\begin{equation}\label{ZNfinal}
	Z_N=CF^{N^2}N^{\frac{1}{12}}\Big(1+O\big(N^{-1}\big)\Big).
\end{equation}
As mentioned before, the $x$ dependency of $C$ will be derived from the Toda equation \eqref{toda}. It implies
\begin{equation}\label{todasimp}
	\frac{d^2}{dx^2}\ln\tau_N = \frac{\tau_{N+1}\tau_{N-1}}{\tau_N^2} = \frac{h_N}{h_{N-1}}
\end{equation}
and we now use equation \eqref{hNalone}
\begin{eqnarray}\label{theo1result}
	h_N&=&\frac{\pi}{2}\bigg(\frac{2\pi N}{\cos\frac{\pi x}{2}}\bigg)^{2N+1}\frac{e^{-2n}}{16^n}\Bigg[1+\frac{1}{4N}+\frac{(-1)^N\cos(\pi x\big(N+\frac{1}{2}\big))}{2N(\ln N)^2}\bigg\{1-\frac{2\ln\ln N}{\ln N}\nonumber\\
	&&+\frac{c_0(x)}{\ln N}+O\Big(\bigg(\frac{\ln\ln N}{\ln N}\bigg)^2\Big)\bigg\}+\frac{c_1(x)}{N^2}+O\bigg(\frac{1}{N^2(\ln N)^2}\bigg)\Bigg],\ N\rightarrow\infty
\end{eqnarray}
uniformly on any compact subset of \eqref{excset}. Here
\begin{equation*}
	c_0(x)=1-2\gamma-4\ln 2-2\ln\Big(\cos\frac{\pi x}{2}\Big)
\end{equation*}
has already been computed, but not $c_1(x)$. Substituting \eqref{theo1result} into \eqref{todasimp}, we get
\begin{eqnarray}\label{conseq1}
	\frac{h_N}{h_{N-1}}&=&\bigg(\frac{\pi}{2\cos\frac{\pi x}{2}}\bigg)^2\frac{N(N-1)}{e^2}\bigg[\Big(1-\frac{1}{N}\Big)^N\bigg]^{-2}\Bigg[1+\frac{(-1)^N\cos\pi xN\ \cos\frac{\pi x}{2}}{N(\ln N)^2}\nonumber\\
	&&\times\bigg\{1-\frac{2\ln\ln N}{\ln N}+\frac{c_0(x)}{\ln N}+O\Big(\bigg(\frac{\ln\ln N}{\ln N}\bigg)^2\Big)\bigg\}\nonumber\\
	&&-\frac{1}{4N^2}+O\bigg(\frac{1}{N^2(\ln N)^2}\bigg)\Bigg],\ \ N\rightarrow\infty.
\end{eqnarray}
Since
\begin{equation*}
	\Big(1+\frac{x}{N}\Big)^N=e^x\Big(1-\frac{x^2}{2N}+\frac{3x^4+8x^3}{24N^2}+O\big(N^{-3}\big)\Big),\ \ N\rightarrow\infty
\end{equation*}
we can simplify \eqref{conseq1} further
\begin{eqnarray}\label{conseq2}
	\frac{h_N}{h_{N-1}}&=&\bigg(\frac{\pi N}{2\cos\frac{\pi x}{2}}\bigg)^2\Bigg[1+\frac{(-1)^n\cos\pi xN\ \cos\frac{\pi x}{2}}{N(\ln N)^2}\bigg\{1-\frac{2\ln\ln N}{\ln N}+\frac{c_0(x)}{\ln N}\nonumber\\
	&&+O\Big(\bigg(\frac{\ln\ln N}{\ln N}\bigg)^2\Big)\bigg\}
	-\frac{1}{12N^2}+O\bigg(\frac{1}{N^2(\ln N)^2}\bigg)\Bigg],\ \ N\rightarrow\infty.
\end{eqnarray}
At this point we notice that
\begin{equation*}
	\bigg(\frac{\pi}{2\cos\frac{\pi x}{2}}\bigg)^2=-\Big(\ln\cos\frac{\pi x}{2}\Big)''
\end{equation*}
as well as
\begin{equation*}
	\bigg(\frac{\pi}{2\cos\frac{\pi x}{2}}\bigg)^2\cos\pi xN\ \cos\frac{\pi x}{2}=\bigg(-\frac{\cos \pi xN}{4N^2\cos\frac{\pi x}{2}}+O\big(N^{-3}\big)\bigg)'',\ \ N\rightarrow\infty
\end{equation*}
which implies with \eqref{conseq2}
\begin{equation*}
	\frac{h_N}{h_{N-1}}=-\Big(\ln\cos\frac{\pi x}{2}\Big)''\Big(N^2-\frac{1}{12}\Big)+O\bigg(\frac{1}{(\ln N)^2}\bigg)
\end{equation*}
and the error term in this identity is uniform for any $x$ chosen from a compact subset of \eqref{excset}. Thus from \eqref{IZconst}
\begin{equation*}
	(\ln Z_N)''=N^2\bigg(\ln\frac{\pi(1-x^2)}{2\cos\frac{\pi x}{2}}\bigg)''+\bigg(\frac{1}{12}\ln N\bigg)''+\frac{1}{12}\Big(\ln\cos\frac{\pi x}{2}\Big)''+O\bigg(\frac{1}{(\ln N)^2}\bigg)
\end{equation*}
and by integration with respect to $x$
\begin{eqnarray}\label{conseq3}
	\ln Z_N &=& N^2\ln \frac{\pi(1-x^2)}{2\cos\frac{\pi x}{2}}+\frac{1}{12}\ln N+\frac{1}{12}\ln\cos\frac{\pi x}{2}+d_1(N)x+d_0(N)\nonumber\\
	&&+O\bigg(\frac{1}{(\ln N)^2}\bigg)
\end{eqnarray}
where $d_0$ and $d_1$ in general depend on $N$, but not on $x$. Substituting \eqref{ZNfinal} into \eqref{conseq3}, we obtain that
\begin{equation}\label{conseq4}
	\ln C = \frac{1}{12}\ln\cos\frac{\pi x}{2}+d_1(N)x+d_2(N)+O\bigg(\frac{1}{(\ln N)^2}\bigg)
\end{equation}
and $C$ does not depend on $N$. For any $x_1,x_2\in(-1,1)$ this implies
\begin{eqnarray*}
	\ln C(x_1)-\ln C(x_2) &=& \frac{1}{12}\ln\cos\frac{\pi x_1}{2}-\frac{1}{12}\ln\cos\frac{\pi x_2}{2}+d_1(N)(x_1-x_2)\\
	&&+O\bigg(\frac{1}{(\ln N)^2}\bigg),
\end{eqnarray*}
i.e. the limit
\begin{equation*}
		\lim_{N\rightarrow\infty}d_1(N) = \kappa_1
\end{equation*}
exists, hence also the limit
\begin{equation*}
	\lim_{N\rightarrow\infty}d_2(N) = \kappa_2
\end{equation*}
exists. By taking the limit $N\rightarrow\infty$ in \eqref{conseq4}, we obtain that
\begin{equation*}
	\ln C = \frac{1}{12}\ln\cos\frac{\pi x}{2}+\kappa_1x+\kappa_2
\end{equation*}
thus proving
\begin{proposition}\label{prop2} The constant $C$ in the asymptotic formula \eqref{ZNfinal} has the form
\begin{equation*}
	C=\bigg(\cos\frac{\pi x}{2}\bigg)^{\frac{1}{12}}e^{\kappa_1x+\kappa_2}.
\end{equation*}
\end{proposition}
We will now show that, in fact, $\kappa_1=0$. To this end recall the initial Izergin-Korepin formula \eqref{IKformula1} and \eqref{IKformula2}, in our situation
\begin{equation*}
	Z_N=\frac{(1-x^2)^{N^2}}{(\prod_{k=0}^{N-1}k!)^2}\tau_N,\ \ \ \tau_N=\det\bigg(\frac{d^{i+j-2}}{dx^{i+j-2}}\phi(x)\bigg)_{i,j=1}^N,\ \ \phi(x)=\frac{2}{1-x^2}.
\end{equation*}
It shows, that $Z_N$ as a function of $x$ is even, hence in the notation of Proposition \ref{prop2}, we conclude
\begin{equation*}
	\kappa_1=0,
\end{equation*}
thus proving Theorem \ref{theo2}.


\section{Phase transition - proof of Theorem \ref{theo3}}

We start with change of variables \eqref{newcoordinates}. It implies in the disordered phase region, via \eqref{Dpara},
\begin{equation*}
	\frac{\sin(\gamma-t)}{\sin(2\gamma)} = \frac{1-x}{2}+y,\qquad \frac{\sin(\gamma+t)}{\sin(2\gamma)}=\frac{1+x}{2}+y,
\end{equation*}
hence for $y>0$,
\begin{equation}\label{gtD}
	\sin\gamma = 2\sqrt{\frac{y(1+y)}{(1+2y)^2-x^2}}\,,\qquad \sin t=x\sin\gamma.
\end{equation}
On the other hand, in the antiferroelectric phase region, via \eqref{AFpara},
\begin{equation*}
	\frac{\sinh(\gamma-t)}{\sinh(2\gamma)}=\frac{1-x}{2}+y,\qquad \frac{\sinh(\gamma+t)}{\sinh(2\gamma)}=\frac{1+x}{2}+y,
\end{equation*}
hence for $y<0$,
\begin{equation}\label{gtAF}
	\sinh\gamma=2\sqrt{\frac{-y(1+y)}{(1+2y)^2-x^2}}\,,\qquad \sinh t=x\sinh\gamma.
\end{equation}
The functions $\sin z$ and $\sinh z$ are both entire, satisfying the usual relations,
\begin{equation*}
	\sin(-z)=-\sin z,\quad \sinh(-z)=-\sinh z,\quad \sinh z=-i\sin(iz).
\end{equation*}
This implies that the inverse functions, $\arcsin z$ and $\textnormal{arcsinh}\ z$, are analytic at $z=0$, and they satisfy
the relations,
\begin{equation*}
	\arcsin(-z)=-\arcsin z,\ \ \textnormal{arcsinh}(-z)=-\textnormal{arcsinh}\ z,\ \ \ \textnormal{arcsinh} z=-i\arcsin(iz).
\end{equation*}
Let us make the change of variable $y=k^2$ in the disordered phase  and $y=-k^2$ in the antiferroelectric phase, where $k>0$. 
Then \eqref{gtD} implies that
\begin{equation}\label{gammaD}
	\gamma=f_{D}(x,k) \equiv \arcsin\Bigg(2k\sqrt{\frac{1+k^2}{(1+2k^2)^2-x^2}}\Bigg),
\end{equation} 
while from \eqref{gtAF},
\begin{equation}\label{gammaAF}
	\gamma=f_{AF}(x,k)\equiv \textnormal{arcsinh}\Bigg(2k\sqrt{\frac{1-k^2}{(1-2k^2)^2-x^2}}\Bigg).
\end{equation}
Here the both functions, $f_D(x,k)$ and $f_{AF}(x,k)$, are analytic at $k=0$ for any $x\in(-1,1)$, satisfying
the relations,
\begin{equation*}
	f_D(x,-k)=-f_D(x,k),\quad f_{AF}(x,-k)=-f_{AF}(x,k),\quad f_{AF}(x,k)=-if_D(x,ik).
\end{equation*}
Now, from \eqref{gtD} and \eqref{gtAF} we have that
\begin{align}
	t&=g_D(x,k) \equiv \arcsin\big(2x\sin\gamma\big)=\arcsin\big(2x\sin f_D(x,k)\big)\label{gteqD}\\
	t&=g_{AF}(x,k)\equiv\textnormal{arcsinh}\big(2x\sinh\gamma\big)=\textnormal{arcsinh}\big(2x\sinh f_{AF}(x,k)\big).\label{gteqAF}
\end{align}
Observe that for any $x\in(-1,1)$, the functions $g_D(x,k),g_{AF}(x,k)$ are both analytic at $k=0$, satisfying
the relations,
\begin{equation*}
	g_D(x,-k)=-g(x,k),\quad g_{AF}(x,-k)=-g_{AF}(x,k),\quad g_{AF}(x,k)=-ig_D(x,ik).
\end{equation*}	
The four equations, \eqref{gammaD}, \eqref{gammaAF}, \eqref{gteqD}, and \eqref{gteqAF}, prove stated behavior 
\eqref{gt} of $\gamma$ and $t$ as $y\rightarrow 0$.

Next we go back to \eqref{energyD} and \eqref{energyAF} and employ change of variables \eqref{newcoordinates},
\begin{equation}
\begin{aligned}
	F_D(x,k^2) &= \frac{\pi ab}{c^2}\frac{\sin(2\gamma)}{2\gamma \cos\frac{\pi t}{2\gamma}},
\quad \gamma=f_D(x,k),\quad t=g_D(x,k)\\
	F_{AF}^{\rm reg}(x,k^2) &= \frac{\pi ab}{c^2}\frac{\sinh(2\gamma)}{2\gamma\cos\frac{\pi t}{2\gamma}},
\quad \gamma=f_{AF}(x,k),\quad t=g_{AF}(x,k).
\end{aligned}
\end{equation}
We first notice that the functions
\begin{equation*}
	h_D(x,k)\equiv\frac{\sin(2\gamma)}{2\gamma} = \frac{\sin 2f_D(x,k)}{2f_D(x,k)},\quad
 h_{AF}(x,k)\equiv \frac{\sinh 2\gamma}{2\gamma}=\frac{\sinh 2f_{AF}(x,k)}{2f_{AF}(x,k)}
\end{equation*}
are both analytic at $k=0$ and
\begin{equation}\label{rel1}
	h_D(-k,x)=h_D(x,k),\ \ \ h_{AF}(-k,x)=h_{AF}(x,k),\ \ \ h_{AF}(x,k)=h_D(ik,x).
\end{equation}
Similarly, the ratios
\begin{equation*}
	\frac{t}{\gamma} = \frac{g_D(x,k)}{f_D(x,k)},\qquad  \frac{t}{\gamma}=\frac{g_{AF}(x,k)}{f_{AF}(x,k)}
\end{equation*}
in the corresponding phase regions, are analytic at $k=0$ satisfying the same relations as in \eqref{rel1},
\begin{equation}\label{rel2}
	\frac{g_D(x,-k)}{f_D(x,-k)}=\frac{g_D(x,k)}{f_D(x,k)},\quad
 \frac{g_{AF}(x,-k)}{f_{AF}(x,-k)}=\frac{g_{AF}(x,k)}{f_{AF}(x,k)},\quad \frac{g_{AF}(x,k)}{f_{AF}(x,k)}=\frac{g_D(x,ik)}{f_D(x,ik)}.
\end{equation}
But this shows that the functions
\begin{equation*}
	r_D(x,k)\equiv \frac{\sin(2\gamma)}{2\gamma\cos\frac{\pi t}{2\gamma}},\qquad
 r_{AF}(x,k)\equiv\frac{\sinh(2\gamma)}{2\gamma\cos\frac{\pi t}{2\gamma}}
\end{equation*}	
are analytic at $k=0$ and they satisfy again the same relations,
\begin{equation*}
	r_D(x,-k)=r_D(x,k),\quad r_{AF}(x,-k)=r_{AF}(x,k),\quad r_{AF}(x,k)=r_D(x,ik),
\end{equation*}	
so in a neighborhood of $k=0$ we have the following convergent Taylor expansions:
\begin{equation}\label{taylor}
	r_D(x,k)=\sum_{j=0}^{\infty}r_j(x)k^{2j},\hspace{0.5cm}r_{AF}(x,k)=\sum_{j=0}^{\infty}(-1)^jr_j(x)k^{2j}.
\end{equation}
Substituting back $y$, and noticing that
\begin{equation*}
	\frac{\pi ab}{c^2} = \Big(\frac{1}{2}+y\Big)^2-\frac{x^2}{4}
\end{equation*}
is clearly analytic at $y=0$, we obtain
\begin{equation}\label{finalseries}
	F_D(x,y) = F_{AF}^{\textnormal{reg}}(x,y) = \sum_{j=0}^{\infty}f_j(x)y^j,\quad x\in(-1,1),
\end{equation}
with
\begin{equation*}
	f_0(x)=\frac{\pi(1-x^2)}{4\cos\frac{\pi x}{2}},\quad
 f_1(x)=\frac{\pi(\pi x^3\sin\frac{\pi x}{2}-\pi x\sin\frac{\pi x}{2}+8\cos\frac{\pi x}{2})}{12\cos^2\frac{\pi x}{2}}\,,
\end{equation*}
and the stated series in \eqref{finalseries} is convergent for small $y$.

Finally, to prove \eqref{singesti}, observe that by \eqref{jac2},
\begin{equation*}
	F^{\textnormal{sing}}_{AF}(\gamma,t)=O(q^2)=O\left(e^{-\frac{\pi^2}{\gamma}}\right).
\end{equation*}

\newpage

\newpage

\end{document}